\RequirePackage{fix-cm}
\documentclass[11pt]{article}
\usepackage[margin=1in]{geometry}
\usepackage{amsmath,amssymb,amsthm,mathtools}
\usepackage{microtype}
\usepackage{enumitem}
\usepackage[normalem]{ulem}
\usepackage[hidelinks]{hyperref}
\usepackage{caption} 
\usepackage{float}
\usepackage{algorithmic}
\usepackage{graphicx,fancyhdr,subfig,url,array,multirow,makecell,xcolor,booktabs,authblk}
\usepackage[linesnumbered,ruled]{algorithm2e}
\usepackage{tikz}
\usepackage{tikz}
\usepackage{graphicx}
\usepackage{colortbl}
\usetikzlibrary{arrows.meta}
\usetikzlibrary{positioning}

\SetAlgorithmName{Rule}{rule}{List of Rules}
\newtheorem{theorem}{Theorem}

\newtheorem{example}{Example}

\newtheorem{definition}{Definition}

\newtheorem{lemma}[theorem]{Lemma}

\newtheorem{proposition}[theorem]{Proposition}

\newcommand{\elloss}{\ell}

\newcommand{\NM}{\mathrm{NM}}

\newcommand{\SNM}{\mathrm{SNM}}
\newcommand{\MNM}{\mathrm{MNM}}
\newcommand{\pos}[1]{\left(#1\right)_{+}}

\newtheorem{ruledef}{Rule}

\newcommand{\BBWS}{%
  \textnormal{\textsc{BlockBonusedWinStrength}}%
}

\title{\textbf{How Well Can Strategyproof Tournament Rules Resist Pairwise Manipulation?}}
\author[1]{Ke Ding}
\author[1]{Bo Li}
\author[1]{Fangxiao Wang}

\affil[1]{Department of Computing,
The Hong Kong Polytechnic University, Hong Kong, China}

\affil[ ]{\small
\texttt{coco-ke.ding@connect.polyu.hk}\\
\texttt{comp-bo.li@polyu.edu.hk}\\
\texttt{fangxiao.wang@connect.polyu.hk}}
\date{}
\begin{document}
\maketitle

\begin{abstract}
A tournament rule maps the outcomes of all pairwise matches among $n$ teams to a possibly randomized winner. Desirable rules should be Condorcet consistent and monotone, yet also resistant to manipulation among coalition. Prior work mostly measures such manipulation additively through $k$-strongly non-manipulable at $\alpha$ ($k$-SNM-$\alpha$), meaning that no coalition of size $k$ can fix the matches among themselves to increase their total winning probability by $\alpha$. Very recently, two new notions of non-manipulability were introduced. Multiplicative non-manipulability ($k$-MNM-$\delta$) is defined analogously, using the multiplicative factor instead. Non-manipulability for $\lambda$ ($k$-NM$_\lambda$) characterizes the selfishness of a team, which restricts a coalition's gain to be less than $\lambda$ times the winning probability sacrificed by its members. 

In this work, we begin with a strict hierarchy among these three notions: NM$_\lambda$ is stronger than MNM, which is then stronger than SNM. This motivates us to consider those two notions that are stronger but less studied: pairwise multiplicative non-manipulability and $2$-non-manipulability for $\lambda$.
We show that Randomized Death Match is $2$-MNM-$3/2$ and optimally matches the lower bound. Then, we introduce the $\BBWS$ rule, which is Condorcet consistent, monotone, and $2$-NM$_2$. This substantially improves the previous upper bound of $\lambda=11$ and comes within a factor of two of the lower bound $\lambda=1$.
\end{abstract}

\section{Introduction}

We study the classic tournament model, where $n$ teams play against each other exactly once. A tournament of size $n$ is represented by a complete oriented graph on $n$ vertices. Each vertex corresponds to a team, and each directed edge from team $i$ to team $j$ means that team $i$ beats team $j$. After observing the outcomes of all $\binom{n}{2}$ matches, the tournament organizer faces the problem of selecting a winner. A deterministic or randomized rule maps every tournament graph to a probability distribution over the $n$ teams.

Tournament rule design is guided by two central criteria: fairness and non-manipulation. We formalize the latter through non-manipulability properties. Fairness properties ensure that some teams satisfying the relevant fairness criterion receive appropriate probabilities of winning the tournament. For example, \emph{Condorcet consistency} requires an undefeated team, if one exists, to be selected with probability one. On the strategic side, a team should always be incentivized to win rather than rewarded for unilaterally throwing a game; we term this requirement \emph{monotonicity}. It is also natural that non-manipulability should prevent a coalition from manipulating its internal matches to improve its members' joint winning probability.

These constraints are motivated by real-world sports scandals. In the women's badminton doubles at the 2012 Olympic Games, four teams that had already secured qualification to the knockout phase attempted to lose on purpose to avoid stronger opponents in the quarterfinals. Consequently, all four teams were disqualified for deliberately attempting to lose. This incident illustrates how a competition format can create incentives that conflict with monotonicity. At the 1982 FIFA World Cup, West Germany's narrow victory over Austria allowed both teams to advance from the group stage at Algeria's expense. The match became known as the ``Disgrace of Gij\'on,'' and the final matches in each group have since been scheduled simultaneously.

Unfortunately, previous work has shown that even if two teams are ``selfless'' and care only about their joint winning probability, Condorcet consistency cannot be achieved together with pairwise strong non-manipulability, or $2$-SNM \cite{DBLP:conf/dagstuhl/AltmanPT10}. Thus, a line of work studied rules that were approximately SNM. A rule is $k$-SNM-$\alpha$ if no coalition of size at most $k$ can fix the matches among themselves so as to increase their joint winning probability by more than $\alpha$. For $k=2$, several Condorcet consistent and monotone rules attain the optimal bound of $1/3$ \cite{DBLP:conf/innovations/SchneiderSW17,DBLP:conf/innovations/SchvartzmanWZZ20,DBLP:conf/wine/DinevW22,DBLP:conf/aldt/MiksanikSS24}. More recently, Pennock et al.~\cite{DBLP:conf/wine/PennockSW25} initiated the study of the multiplicative version (MNM) by replacing the additive bound $\alpha$ with a multiplicative factor $\delta$, and provided a $2$-MNM-$7/2$ rule.

On the other hand, to characterize the degree of utility transferability, Pennock et al.~\cite{DBLP:conf/aldt/PennockSX24} first introduced a selfishness parameter $\lambda$. One can imagine that $\lambda=0$ represents complete selflessness and $\lambda=\infty$ represents complete selfishness. In this model, one team may be willing to lose on purpose only if the colluding pair's gain in joint winning probability is more than $\lambda$ times the loss in the sacrificing team's winning probability. Although they conjectured that a $2$-NM$_1$ rule exists, all previously known rules require $\lambda\ge\Omega(n)$ to satisfy $2$-NM$_\lambda$. Pennock et al.~\cite{DBLP:conf/wine/PennockSW25} were the first to provide a rule with a constant parameter $\lambda=11$. 

As mentioned, previous work left two substantial gaps that are our major objectives: 
\begin{itemize}
    \item For $2$-MNM-$\delta$, a gap between the upper bound $\delta=7/2$ and lower bound $\delta=3/2$;
    \item For $2$-NM$_\lambda$, a gap between the upper bound $\lambda=11$ and lower bound $\lambda=1$.
\end{itemize}

\subsection{Contribution}
Our contributions are threefold. First, we establish a hierarchy among $k$-SNM-$\alpha$, $k$-MNM-$\delta$, and $k$-NM$_\lambda$. The hierarchy includes the parameter conversions stated below, and both converse implications fail even when $k=2$.

\smallskip

\noindent \textbf{Hierarchical Comparison (Proposition~\ref{the:hierarchy}).} Every $k$-NM$_\lambda$ tournament rule is also $k$-MNM-$(1+\lambda)$ and, consequently, $k$-SNM-$\lambda/(1+\lambda)$.

Although almost all previous work studied the $k$-SNM-$\alpha$ notion, this hierarchy indicates that the two other non-manipulability notions are stronger and more demanding than $k$-SNM-$\alpha$. Next, we tighten the upper bound on multiplicative pairwise non-manipulability ($2$-MNM) and substantially improve the upper bound on pairwise non-manipulability for $\lambda$ ($2$-NM$_\lambda$). We show our improvements in Table \ref{tab:summary-bounds}.

\begin{table}[!tb]
    \centering
    \renewcommand{\arraystretch}{1.3}
    \setlength{\tabcolsep}{7pt}
    \begin{tabular}{lcc}
        \toprule
        \textbf{Guarantee}
        & \textbf{Upper bound}
        & \textbf{Lower bound} \\
        \midrule
        $2$-MNM-$\delta$
        & $\frac72\to \boldsymbol{\frac32}^{*}$ (Thm.~\ref{the:2mnm})
        & $\frac32$ \\
        $2$-NM$_{\lambda}$
        & $11\to \mathbf{2}$ (Thm.~\ref{the:2NM})
        & $1$ \\
        \bottomrule
    \end{tabular}
    \caption{Summary of quantitative bounds. The improved upper bounds are in bold and the asterisk denotes an optimal bound. Smaller values of $\delta$ and $\lambda$ give stronger guarantees.}
    \label{tab:summary-bounds}
\end{table}

As the first main result, we propose the first rule known to satisfy $2$-MNM-$3/2$, which is the optimal bound under this notion. Randomized Death Match is a recursive rule that repeatedly and randomly selects a pair of surviving teams, eliminating the loser of their recorded match; the last surviving team wins.

\smallskip

\noindent \textbf{Main Result 1 (Theorem~\ref{the:2mnm}).} The Randomized Death Match rule is Condorcet consistent, monotone, and $2$-MNM-$3/2$ for every tournament. The factor $3/2$ is optimal among Condorcet consistent tournament rules.

Our second main result substantially improves the current $\lambda=11$ upper bound for $2$-NM$_\lambda$. The analysis of Mik\v{s}an\'{\i}k et al.~\cite{DBLP:conf/aldt/MiksanikSS24} suggests that teams with exactly one loss, which we call \emph{almost-Condorcet winners}, play a central role in pairwise manipulation. Our rule is also inspired by this idea. Let $H(T)$ denote the set of almost-Condorcet winners in tournament $T$. A team that defeats an almost-Condorcet winner is its \emph{blocker}, because reversing that match would create a Condorcet winner. Three possible situations are as follows:
\begin{itemize}
    \item If there is a unique almost-Condorcet winner, the team that defeats it is its blocker.
    \item If there are three almost-Condorcet winners, each defeats one of the other two and loses to the other; each is therefore a blocker.
    \item If there are two almost-Condorcet winners, one defeats the other. The winner of that match is the \emph{internal blocker}, and the team that defeats the internal blocker is the \emph{external blocker}.
\end{itemize}

The $\BBWS$ rule begins with a geometric score determined by each team's win count and adds the scores of the teams it defeats. Intuitively, we make the blockers more resistant to manipulation by assigning them bonus scores, thereby increasing their winning probability. By choosing the bonus scores properly, the proof carefully tracks how blocker roles and total bonuses change.

\noindent \textbf{Main Result 2 (Theorem~\ref{the:2NM}).} The $\BBWS$ rule is Condorcet consistent, monotone, and $2$-NM$_2$ for every tournament. Moreover, $\lambda=2$ is tight for this rule. 

\subsection{Related Work}

The study of strategic manipulation in tournaments was initiated by Altman et al.~\cite{DBLP:conf/dagstuhl/AltmanPT10} and Altman and Kleinberg~\cite{DBLP:conf/aaai/AltmanK10}. Altman et al. proved that no deterministic rule can simultaneously satisfy Condorcet consistency and pairwise strong non-manipulability \cite{DBLP:conf/dagstuhl/AltmanPT10}. Altman and Kleinberg then studied randomized tournament rules and obtained positive results by replacing Condorcet consistency with weaker fairness requirements while retaining monotonicity and $2$-SNM-$\alpha$ \cite{DBLP:conf/aaai/AltmanK10}.

Since it is unimaginable to not award the winner to an undefeated team, subsequent work shifted towards relaxing $2$-SNM-$\alpha$ while maintaining Condorcet consistency and monotonicity. For pairwise collusion, several simple rules---including Randomized Single Elimination Bracket, Randomized King-of-the-Hill, and Randomized Death Match---are Condorcet consistent, monotone, and $2$-SNM-$1/3$; the additive constant $1/3$ is optimal \cite{DBLP:conf/innovations/SchneiderSW17,DBLP:conf/innovations/SchvartzmanWZZ20,DBLP:conf/wine/DinevW22,DBLP:conf/aldt/MiksanikSS24}. Much less is known for coalitions of size $k>2$. Schneider et al.~\cite{DBLP:conf/innovations/SchneiderSW17} proved the lower bound $\alpha\ge (k-1)/(2k-1)$. Mik\v{s}an\'{\i}k et al.~\cite{DBLP:conf/aldt/MiksanikSS24} later constructed a Condorcet consistent, monotone, $3$-SNM-$1/2$ rule and the first explicit family of Condorcet consistent, monotone, $k$-SNM-$\alpha$ rules with $\alpha<1$ for every $k$. Most recently, Pennock et al. introduced the notions of $2$-MNM and $2$-NM$_\lambda$ \cite{DBLP:conf/aldt/PennockSX24,DBLP:conf/wine/PennockSW25}, which are our main focus in this paper.  

Two papers also extend the standard tournament model: Ding and Weinberg~\cite{DBLP:conf/innovations/DingW21} studied probabilistic match outcomes, while Dale et al.~\cite{DBLP:conf/sigecom/DaleFRSW22} generalized winner-take-all tournaments to multiple-prize settings. Beyond coalitional non-manipulability, related work has examined other aspects of tournament design. In particular, several tournament rules have been studied extensively \cite{d994415c-d3ba-3116-95d8-3fd339d786d1, 10.1007/978-3-540-77105-0_30,SANVER2010354, 0499245f-096c-3e3d-a230-e8eb0387217a,repec:pra:mprapa:93006}.
A separate line of work on the \emph{tournament fixing problem} asks how an organizer can choose a single-elimination bracket to favor a desired outcome. We refer the survey of Suksompong~\cite{DBLP:conf/ijcai/Suksompong21}. 

\section{Preliminaries}

Throughout the paper, we consider complete round-robin tournaments. This section fixes the notation for tournaments and tournament rules, states the requirements of Condorcet consistency and monotonicity, and defines the three notions of coalitional non-manipulability compared and analyzed in the following sections.

\begin{definition}[Tournament]
A tournament $T$ on a set of teams $N=[n]=\{1,\ldots,n\}$ is the outcome of the $\binom{n}{2}$ pairwise matches between the teams. Equivalently, $T$ is an orientation of the complete graph $K_n$. We write $i\to_T j$ if team $i$ defeats team $j$ in $T$. The set of all tournaments on $N$ is denoted by $\mathcal{T}_n$. For simplicity, let $d_i(T)=|\{j\in N\setminus\{i\}:i\to_T j\}|$ be the outdegree (or the number of wins) of team $i$.
\end{definition}

\smallskip

\begin{definition}[Tournament rule]
A randomized tournament rule is a map $r:\mathcal{T}_n\longrightarrow\Delta_n$, where $\Delta_n$ denotes the probability simplex over $N$. For each $T\in\mathcal{T}_n$ and $i\in N$, $r_i(T)$ is the probability that the rule selects team $i$ as the winner of $T$.
\end{definition}

Condorcet consistency requires that a team that is undefeated must be selected as the winner.

\begin{definition}[Condorcet consistency]
A team $i$ is a Condorcet winner of $T$ if $d_i(T)=n-1$, meaning that $i$ beats every other team. A tournament rule $r$ is Condorcet consistent if $r_i(T)=1$ whenever $i$ is a Condorcet winner of $T$.
\end{definition}

\smallskip

\begin{definition}[$S$-adjacent tournaments]
Let $S\subseteq N$. Two tournaments $T,T'\in\mathcal{T}_n$ are $S$-adjacent if every match whose two endpoints are not both in $S$ has the same outcome in $T$ and $T'$. In particular, two $\{i,j\}$-adjacent tournaments differ only in the outcome of the match between $i$ and $j$.
\end{definition}

\smallskip

Naturally, monotonicity requires that a team cannot increase its winning probability by intentionally throwing a match that it could have won.

\begin{definition}[Monotonicity]
A tournament rule $r$ is monotone if, for every pair of $\{i,j\}$-adjacent tournaments $T,T'$ such that $i\to_T j$ and $j\to_{T'} i$, one has $r_i(T)\ge r_i(T')$.
\end{definition}

\smallskip

The next two notions formalize coalitional non-manipulability by imposing, respectively, additive and multiplicative bounds on the increase in a coalition's joint winning probability. For the rest of the paper, denote $\sum_{i\in S} r_i(T)=r_S(T)$ for every tournament $T$ and team set $S$. 

\begin{definition}[$k$-SNM-$\alpha$]
Let $1\leq k\leq n$ and $0\le \alpha \le 1$. A tournament rule
$r:\mathcal{T}_n\to\Delta_n$ is \emph{$k$-strongly
non-manipulable up to $\alpha$}, abbreviated $k$-SNM-$\alpha$,
if, for every nonempty coalition $S\subseteq N$ with $|S|\leq k$
and every ordered pair of $S$-adjacent tournaments $T,T'$, $r_S(T')\leq r_S(T)+\alpha$.
\end{definition}

\begin{definition}[$k$-MNM-$\delta$]
Let $k\ge1$ and $\delta\ge1$. A tournament rule $r$ is $\delta$-multiplicative $k$-non-manipulable, abbreviated $k$-MNM-$\delta$, if, for every coalition $S\subseteq N$ with $|S|\le k$ and every ordered pair of $S$-adjacent tournaments $T,T'$, $r_S(T')\le\delta\cdot r_S(T)$.

For $k=2$, this becomes
\[
r_i(T')+r_j(T')\le\delta\bigl(r_i(T)+r_j(T)\bigr).
\]
\end{definition}

Equivalently, no coalition of size at most $k$ can increase its
joint winning probability by more than $\alpha$ under $k$-SNM-$\alpha$ (more than $\delta$ times under $k$-MNM-$\delta$, resp.) by fixing matches among its members. The $k=2$ case is our specific concern in Section \ref{sec:mnm}.

\vspace{2pt}

The following definition models collusion under partially transferable utilities \cite{DBLP:conf/aldt/PennockSX24}. Although the $k=2$ case is our main focus in Section \ref{sec:2NM}, we state the notion for coalitions of size at most $k$ because the general form is used in Section \ref{sec:hierarchy}. Informally, the increase in a coalition's joint winning probability is bounded by $\lambda$ times the sum of its members' sacrifices in winning probability. 

\begin{definition}[\(k\)-\(\mathrm{NM}_{\lambda}\)]
\label{def:k-nm}
Let \(k\ge 2\) and \(\lambda\ge 0\). A tournament rule \(r\) is
\emph{\(k\)-non-manipulable for \(\lambda\)}, abbreviated
\(k\)-\(\mathrm{NM}_{\lambda}\), if, for every coalition
\(S\subseteq N\) with \(|S|\le k\) and every ordered pair of
\(S\)-adjacent tournaments \(T,T'\),
\[
r_S(T')
\le
r_S(T)
+
\lambda
\sum_{i\in S}
\pos{r_i(T)-r_i(T')},
\]
where \(\pos{x}:=\max\{x,0\}\).
\end{definition}

For $k=2$, this definition coincides with the $2$-NM$_\lambda$ definition of Pennock et al.\ \cite{DBLP:conf/aldt/PennockSX24}. Explicitly, for every pair of distinct teams $i,j$ and every ordered pair of $\{i,j\}$-adjacent tournaments $T,T'$,
\[
\begin{aligned}
r_i(T')+r_j(T')\le{}r_i(T)+r_j(T)+\lambda\max\{r_i(T)-r_i(T'),\,r_j(T)-r_j(T')\}.
\end{aligned}
\]
Pennock et al.\ \cite{DBLP:conf/aldt/PennockSX24} also show that, for a monotone rule, the $2$-NM$_\lambda$ condition is equivalent to the following gain--loss formulation. For $\{i,j\}$-adjacent tournaments $T,T'$ with $i\to_T j$ and $j\to_{T'}i$,
\begin{equation}
\label{eq:gain-loss}
r_j(T')-r_j(T)\le(\lambda+1)\bigl(r_i(T)-r_i(T')\bigr).
\end{equation}
In other words, $2$-NM$_\lambda$ implies that when team $i$ throws a match to team $j$, $j$'s increase in winning probability is at most $\lambda+1$ times $i$'s decrease.

\smallskip

The following directed-triangle example in Figure~\ref{fig:triangle} shows that for every Condorcet consistent rule, the best possible guarantees cannot be better than $\alpha\ge 1/3$ for $2$-SNM-$\alpha$, $\delta\ge 3/2$ for $2$-MNM-$\delta$, and $\lambda\ge 1$ for $2$-NM$_\lambda$ \cite{DBLP:conf/innovations/SchneiderSW17, DBLP:conf/wine/PennockSW25, DBLP:conf/aldt/PennockSX24} .

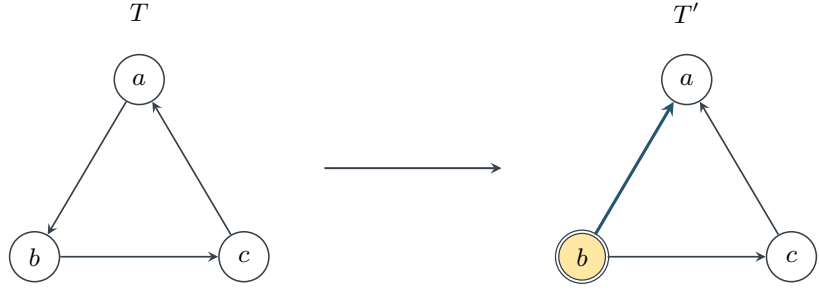
\begin{figure}[t]
\centering
\begingroup
\hyphenpenalty=10000
\exhyphenpenalty=10000

\definecolor{triangleblue}{RGB}{36,87,110}
\definecolor{triangleline}{RGB}{55,65,75}
\definecolor{trianglegold}{RGB}{255,232,163}

\resizebox{.80\linewidth}{!}{%
\begin{tikzpicture}[
  team/.style={
    circle,
    draw=triangleline,
    fill=white,
    line width=.55pt,
    minimum size=6.2mm,
    inner sep=0pt,
    font=\footnotesize
  },
  cw team/.style={
    team,
    fill=trianglegold,
    double,
    double distance=.7pt,
    line width=.45pt
  },
  match/.style={->,>=stealth,line width=.6pt,draw=triangleline},
  changed match/.style={->,>=stealth,line width=1.1pt,draw=triangleblue},
  transition/.style={->,>=stealth,line width=.7pt,draw=triangleline},
  every node/.style={font=\footnotesize}
]

\path[use as bounding box] (0,0) rectangle (12.4,4.4);

\node[font=\footnotesize\bfseries] at (2.8,4.08) {$T$};
\node[font=\footnotesize\bfseries] at (9.6,4.08) {$T'$};

\node[team] (aT) at (2.8,3.25) {$a$};
\node[team] (bT) at (1.5,1.05) {$b$};
\node[team] (cT) at (4.1,1.05) {$c$};
\draw[match] (aT) -- (bT);
\draw[match] (bT) -- (cT);
\draw[match] (cT) -- (aT);

\node[team] (aTp) at (9.6,3.25) {$a$};
\node[team] (bTp) at (8.3,1.05) {$b$};
\node[cw team] (bTp) at (8.3,1.05) {$b$};
\node[team] (cTp) at (10.9,1.05) {$c$};
\draw[changed match] (bTp) -- (aTp);
\draw[match] (bTp) -- (cTp);
\draw[match] (cTp) -- (aTp);

\draw[transition] (5.1,2.15) -- (7.3,2.15);

\end{tikzpicture}%
}
\endgroup
\caption{A tournament which attains the three lower bounds for all tournament rules. Reversing any edge makes a team Condorcet winner.}
\label{fig:triangle}
\end{figure}

\begin{example}[Lower bounds]
\label{prop:triangle-lower-bounds}
Consider the tournament $T$ on three teams $a,b,c$ satisfying $a\to_T b, b\to_T c, c\to_T a$. Then no Condorcet consistent tournament rule is $2$-SNM-$\alpha$ for any $\alpha<1/3$, no Condorcet consistent tournament rule is $2$-MNM-$\delta$ for any $\delta<3/2$, and no Condorcet consistent tournament rule is $2$-NM$_\lambda$ for any $\lambda<1$.
\end{example}

Having defined the three notions, we now establish their hierarchy by determining which notions imply the others in the next section.

\section{A Hierarchy of Coalitional Non-Manipulability Notions}
\label{sec:hierarchy}

The three notions introduced above bound the increase in a coalition's joint winning probability in different ways. The additive condition ($k$-SNM-$\alpha$) bounds this increase by a fixed constant, the multiplicative condition ($k$-MNM-$\delta$) bounds it by a multiplied ratio, and the $k$-$\NM_\lambda$ bounds it relative to the sum of winning probability sacrificed by each member. Accordingly, we uncover the relationship among the three guarantees by determining the implications among them.

For a tournament rule $r$, a coalition $S\subseteq N$, and $S$-adjacent tournaments $T,T'$, denote $L_S(T,T'):=\sum_{i\in S}\pos{r_i(T)-r_i(T')}$ by the sum of each coalition member's positive decrease in winning probability.

\begin{proposition}
\label{the:hierarchy}
For every $k\ge 2$ and every $\lambda\ge 0$,
\[
  k\text{-}\NM_{\lambda}
  \;\Longrightarrow\;
  k\text{-}\MNM\text{-}(1+\lambda)
  \;\Longrightarrow\;
  k\text{-}\SNM\text{-}\frac{\lambda}{1+\lambda}.
\]
Neither implication requires Condorcet consistency or monotonicity.
\end{proposition}

\begin{proof}
Fix a tournament rule $r$, an arbitrary coalition $S\subseteq N$ with $|S|\le k$, and an arbitrary ordered pair of $S$-adjacent tournaments $T,T'$. Suppose first that $r$ is $k$-$\NM_{\lambda}$. Then $r_S(T')\le r_S(T)+\lambda L_S(T,T')$.
To bound $r_S(T')$ by $(1+\lambda)r_S(T)$, note that each member can lose at most its initial winning probability: for every $i\in S$, non-negativity gives $\pos{r_i(T)-r_i(T')}\le r_i(T)$. Summing over $i\in S$ yields $L_S(T,T')\le r_S(T)$.
Hence, we can obtain $r_S(T')\le r_S(T)+\lambda r_S(T)=(1+\lambda)r_S(T)$. Since $S,T,T'$ were arbitrary, $r$ is $k$-$\MNM$-$(1+\lambda)$.

For the second implication, we convert a multiplicative bound into an additive bound using the fact that the coalition's winning probability is at most one. It is noted that any rule tight on $k$-MNM-$\delta$ must also be tight on $k$-SNM-$\alpha$.
Suppose that $r$ is $k$-$\MNM$-$\delta$ with $\delta\ge 1$. Then $r_S(T)\ge r_S(T')/\delta$. Since $r_S(T')\le 1$ and $\delta\ge 1$, it follows that
\[
  r_S(T')-r_S(T)
  \le \left(1-\frac{1}{\delta}\right)r_S(T')
  \le\frac{\delta-1}{\delta}.
\]
Since $S,T,T'$ were arbitrary, the last inequality shows that $r$ is $k$-$\SNM$-$(\delta-1)/\delta$.
Taking $\delta=1+\lambda$ gives the second implication.
\end{proof}

However, the converse implications in Proposition~\ref{the:hierarchy} fail in general, even for monotone rules and coalitions of size two, as the following examples show. Example~\ref{ex:snm-not-mnm} shows that an additive guarantee can allow a coalition to increase its winning probability from zero, whereas a multiplicative guarantee cannot.

\begin{example}[Observation 4.3 in \cite{DBLP:conf/wine/PennockSW25}]
\label{ex:snm-not-mnm}
The Randomized Single Elimination Bracket rule is Condorcet consistent, monotone, and $2$-SNM-$1/3$, but it is not $2$-MNM-$\delta$ for any $\delta$.
\end{example}

Moreover, example~\ref{ex:mnm-not-nm} shows that the reverse implication from multiplicative non-manipulability to $2$-$\NM_\lambda$ fails for any $\lambda=o(n)$.

\begin{example}[Theorem 4.4 and Observation 5.2 in \cite{DBLP:conf/wine/PennockSW25}]
\label{ex:mnm-not-nm}
The Normalized Geometric Win Count Score rule is Condorcet consistent, monotone, and $2$-MNM-$7/2$, but it is not $2$-NM$_\lambda$ for any $\lambda=o(n)$.
\end{example}

This hierarchy motivates the next two rule-design objectives: attaining the optimal pairwise multiplicative guarantee and reducing the selfishness parameter required for pairwise non-manipulability. Sections~\ref{sec:mnm} and~\ref{sec:2NM} address these objectives, respectively.

\section{A Tight Multiplicative Non-Manipulability Bound}\label{sec:mnm}

We revisit Randomized Death Match (RDM), which is known to be Condorcet consistent, monotone, and $2$-SNM-$1/3$~\cite{DBLP:conf/wine/DinevW22}.
We show that it also attains the optimal guarantee for pairwise multiplicative non-manipulability. 

\begin{samepage}
\begin{ruledef}[Randomized Death Match]
\label{rule:rdm}
{Given a tournament $T$ on a team set $N$, proceed as follows.}
\begin{enumerate}[label=\arabic*.]
\item Initialize the set of surviving teams as $N_0=N$.
\item At each step $t$ with $|N_t|\ge 2$, select a pair $\{i,j\}\subseteq N_t$ uniformly at random. If $i\to_T j$, eliminate $j$ by setting $N_{t+1}=N_t\setminus\{j\}$; otherwise, eliminate $i$ by setting $N_{t+1}=N_t\setminus\{i\}$.
\item When a single team $w$ remains, select $w$ as the winner. Equivalently, $r_i(T)$ is the probability that team $i$ is the last surviving team.
\end{enumerate}
\end{ruledef}
\end{samepage}

\begin{theorem}
\label{the:2mnm}
For every $n\ge 2$, RDM is $2$-MNM-$3/2$. In particular, for every pair $S=\{i,j\}$ and every ordered pair of $S$-adjacent tournaments $T,T'$ on the same vertex set,
\[
    r_i(T')+r_j(T')\le \frac32\bigl(r_i(T)+r_j(T)\bigr).
\]
Moreover, the factor $3/2$ is optimal among all Condorcet consistent tournament rules.
\end{theorem}

To support the induction, we first express a coalition’s joint winning probability in an \(n\)-team tournament in terms of the corresponding probabilities after one team is eliminated. Denote the induced sub-tournament on $N\setminus\{v\}$ by $T-v$ for $v\in N$. Because only surviving coalition members can contribute to $r_S(T)$, we extend the notation $r_S(T)=\sum_{u\in S\cap N}r_u(T)$. Also let $\elloss_v(T)=\bigl|\{u\in N:u\to_T v\}\bigr|$ denote the number of losses of $v$ in $T$. We first show how the joint winning probability of tournament of size $n$ relates to those of size $n-1$.

\begin{lemma}
\label{lem:rdm-recurrence}
For every tournament $T$ with team set $N$ of size $n\ge 2$ and every set $S$ of teams,
\begin{equation}
\label{eq:rdm-recurrence}
    \binom{n}{2}r_S(T)=\sum_{v\in N}\elloss_v(T)\,r_S(T-v).
\end{equation}
\end{lemma}

\begin{proof}
A team $v\in N$ is eliminated first precisely when the selected pair consists of $v$ and one of the $\elloss_v(T)$ teams that defeat it. Since each pair of teams is chosen uniformly at random, this event has probability $\elloss_v(T)/\binom{n}{2}$. After $v$ is eliminated, RDM continues on $T-v$, where the coalition wins with probability $r_S(T-v)$. Summing over $v\in N$ gives \eqref{eq:rdm-recurrence}.
\end{proof}

\begin{lemma}[]
\label{lem:rdm-2mnm}
For every $n\ge 2$, every pair $S=\{i,j\}$ and $S$-adjacent tournaments $T,T'$ on $n$ vertices, $r_S(T')\le \frac32r_S(T).$
\end{lemma}

\begin{proof}
Apply induction on $n$. For $n=2$, the pair $S$ contains every team, so $r_S(T)=r_S(T')=1$.

Fix $n\ge 3$ and assume that the bound holds on fewer than $n$ teams. Fix a pair $S=\{i,j\}$ and $S$-adjacent tournaments $T,T'$ on a common team set $N$ of size $n$. Assuming without loss of generality that $i\to_T j$ and $j\to_{T'} i$. Our goal then is to show $\frac32r_S(T)-r_S(T')\ge 0$ by comparing the contributions of each possible first elimination. Reversing the mutual match adds one loss to $i$ and removes one loss from $j$, so $\elloss_i(T')=\elloss_i(T)+1, \elloss_j(T')=\elloss_j(T)-1$.
Deleting either member of $S$ removes the only modified edge. The coalition's winning probabilities after deleting $i$ and $j$, respectively, are therefore the same in the two tournaments; denote them by
\[
    A=r_j(T-i)=r_j(T'-i),
    \qquad
    B=r_i(T-j)=r_i(T'-j).
\]
Every team outside $S$ has the same number of losses in $T$ and $T'$. Separating deletions within $S$ from deletions outside $S$, Lemma~\ref{lem:rdm-recurrence} gives
\begin{equation}
\label{eq:recurrence-T}
    \binom{n}{2}r_S(T)
    =\elloss_i(T)A+\elloss_j(T)B+\sum_{v\notin S}\elloss_v(T)r_S(T-v)
\end{equation}
and
\begin{equation}
\label{eq:recurrence-Tprime}
    \binom{n}{2}r_S(T')
    =(\elloss_i(T)+1)A+(\elloss_j(T)-1)B+\sum_{v\notin S}\elloss_v(T)r_S(T'-v).
\end{equation}
Subtracting \eqref{eq:recurrence-Tprime} from $3/2$ times \eqref{eq:recurrence-T} yields
\begin{equation}
\label{eq:master-rdm}
\begin{aligned}
    \binom{n}{2}\left(\frac32r_S(T)-r_S(T')\right)
    &=\left(\frac{\elloss_i(T)}{2}-1\right)A
      +\left(\frac{\elloss_j(T)}{2}+1\right)B\\
    &\quad+\sum_{v\notin S}\elloss_v(T)
      \left(\frac32r_S(T-v)-r_S(T'-v)\right).
\end{aligned}
\end{equation}
For every $v\in N\setminus S$, the tournaments $T-v$ and $T'-v$ still contain both members of $S$ and differ only in their mutual match. They are therefore $S$-adjacent on $n-1$ vertices, and for every $v\notin S$, the induction hypothesis gives
\begin{equation*}
    r_S(T'-v)\le \frac32r_S(T-v).
\end{equation*}
Thus every summand in the final sum in \eqref{eq:master-rdm} is nonnegative. The $B$-term is also nonnegative: $B\ge 0$, and $\elloss_j(T)\ge 1$ because $j$ loses to $i$ in $T$. Only the $A$-term remains to be discussed. We distinguish three cases according to $\elloss_i(T)$.

\medskip
\noindent\emph{Case 1: $\elloss_i(T)=0$.}
Team $i$ is a Condorcet winner of $T$, so Condorcet consistency gives $r_i(T)=r_S(T)=1$. Hence $r_S(T')\le 1\le \frac32r_S(T)$.

\medskip
\noindent\emph{Case 2: $\elloss_i(T)\ge 2$.}
Here $\elloss_i(T)/2-1\ge 0$, so every term on the right-hand side of \eqref{eq:master-rdm} is nonnegative. This gives $r_S(T')\le \frac32r_S(T)$.

\medskip
\noindent\emph{Case 3: $\elloss_i(T)=1$.}
Let $x$ be the unique team that defeats $i$ in $T$. Since $i\to_T j$, we have $x\notin S$. If $x$ is a Condorcet winner of $T$, it remains one in $T'$ because the modified match is not incident to $x$. Condorcet consistency then gives $r_S(T)=r_S(T')=0$, proving the bound. Otherwise, $\elloss_x(T)\ge 1$.

Since $\elloss_i(T)=1$, the $A$-term in \eqref{eq:master-rdm} is $-A/2\ge -1/2$, because $0\le A\le 1$. It therefore suffices to show that the term corresponding to the deletion of $x$ contributes at least $1/2$. Deleting $x$ removes the unique loss of $i$, so Condorcet consistency gives $r_S(T-x)=1$. Together with $r_S(T'-x)\le 1$ and $\elloss_x(T)\ge 1$, this yields
\begin{equation}
\label{eq:blocker-slack}
    \elloss_x(T)\left(\frac32r_S(T-x)-r_S(T'-x)\right)
    \ge \elloss_x(T)\left(\frac32-1\right)
    \ge \frac12.
\end{equation}
By \eqref{eq:blocker-slack}, the $A$-term and the $x$-summand have nonnegative sum. All remaining terms in \eqref{eq:master-rdm} are nonnegative, so $r_S(T')\le \frac32r_S(T)$.

Cases 1--3 exhaust the possible values of the nonnegative integer $\elloss_i(T)$. This completes the induction and proves the bound for every $n\ge 2$.
\end{proof}

\begin{proof}[Proof of Theorem~\ref{the:2mnm}]
Lemma~\ref{lem:rdm-2mnm} proves the multiplicative bound for every pair, and coalitions with fewer than two members cannot change any match. Hence RDM is $2$-MNM-$3/2$. Combining this bound with the known Condorcet consistency and monotonicity of RDM gives the stated guarantees. Optimality of the factor $3/2$ is shown in Example \ref{prop:triangle-lower-bounds}.
\end{proof}

\section{Non-Manipulability for \texorpdfstring{$\lambda=2$}{lambda=2}}\label{sec:2NM}

We next turn to the upper bound for $2$-NM$_\lambda$. At a high level, the rule first assigns each team a score based on its win count and the win counts of the teams it defeats. We then augment these scores with bonuses in tournaments containing \emph{almost-Condorcet winners}, that is, teams with exactly one loss. It is easy to see there are at most three such teams. Bonuses are assigned to teams that defeat almost-Condorcet winners, which we term \emph{blockers}, and the bonus values are depend on the number of almost-Condorcet winners.

\begin{ruledef}[$\BBWS$]
\label{rule:2NM}
Given a tournament $T$ on a team set $N$ of size $n$, fix the parameters
\[
 B=\frac{89}{36},
 P=\frac{13}{6},
 Q=\frac{17}{12},
 R=\frac76,
 M=\frac{143}{12},
\]
and proceed as follows.
\begin{enumerate}[label=\arabic*.]
\item If $T$ has a Condorcet winner $w$, set $r_w(T)=1$ and
 $r_i(T)=0$ for every $i\ne w$, and stop.
\item Otherwise, for each $i\in N$, define the score $s_i(T)=3^{d_i(T)-(n-2)}$
and the unbonused win-strength score $t_i^0(T)=\frac32s_i(T)+\sum_{v:i\to_T v}s_v(T)$.

\item Let $H(T)=\{h\in N:d_h(T)=n-2\}$ be the set of almost-Condorcet winners. Define the bonuses $b_i(T)$ according to the size of $H(T)$ as follows:
\begin{enumerate}[label=(\alph*)]
\item If $H(T)=\varnothing$, every bonus is zero.
\item If $H(T)=\{h\}$, the unique blocker of $h$ receives bonus $B$;
all other bonuses are zero.
\item If $H(T)=\{a,b\}$ and $a\to_T b$, the internal blocker $a$
 receives bonus $P$, and the external blocker of $a$ receives bonus $Q$;
all other bonuses are zero.
\item If $|H(T)|=3$, every member of $H(T)$ receives bonus $R$;
all other bonuses are zero.
\end{enumerate}
\item For each team $i\in N$, denote by $t_i(T)=t_i^0(T)+b_i(T)$ its bonused score, and define
its winning probability by
\[
 r_i(T)
 =
 \frac{t_i(T)}{M}
 +\frac1n\left(1-\frac{\sum_{v\in N}t_v(T)}{M}\right).
\]
\end{enumerate}
\end{ruledef}

\begin{theorem}\label{the:2NM}
    The $\BBWS$ rule is Condorcet consistent, monotone, and $2$-NM$_2$ for every tournament. Moreover, the $2$-NM$_2$ guarantee is tight for this rule.
\end{theorem}

\noindent\emph{Proof overview.}
After adding the bonuses to related teams, the main target for the bound $\lambda=2$ is to compare the loss of $i$ with the gain of $j$. The gain–loss analysis splits into three cases: bonuses do not change; bonuses change but no Condorcet winner is created; or bonuses change and a new Condorcet winner is created.” The complete proof flow is illustrated in Figure~\ref{fig:roadmap}.

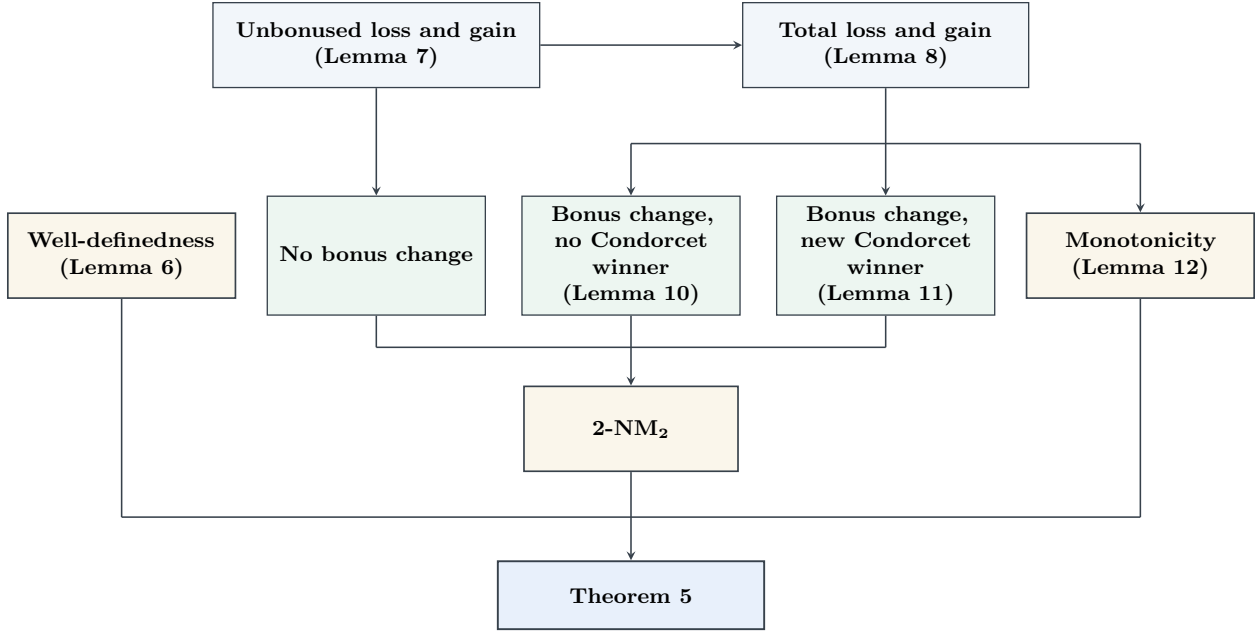
\begin{figure}[t]
\centering
\begingroup
\hyphenpenalty=10000
\exhyphenpenalty=10000

\definecolor{secfiveinput}{RGB}{242,246,250}
\definecolor{secfivecase}{RGB}{237,246,241}
\definecolor{secfiveresult}{RGB}{250,246,235}
\definecolor{secfivetheorem}{RGB}{232,240,251}
\definecolor{secfiveline}{RGB}{55,65,75}

\resizebox{\linewidth}{!}{%
\begin{tikzpicture}[
  box/.style={
    draw=secfiveline,
    rounded corners=0pt,
    fill=secfiveinput,
    line width=.55pt,
    align=center,
    font=\footnotesize\bfseries\boldmath,
    text width=3.25cm,
    minimum height=1.25cm,
    inner sep=3pt
  },
  casebox/.style={
    draw=secfiveline,
    rounded corners=0pt,
    fill=secfivecase,
    line width=.55pt,
    align=center,
    font=\footnotesize\bfseries,
    text width=3.25cm,
    minimum height=1.75cm,
    inner sep=3pt
  },
  resultbox/.style={
    draw=secfiveline,
    rounded corners=0pt,
    fill=secfiveresult,
    line width=.7pt,
    align=center,
    font=\footnotesize\bfseries\boldmath,
    text width=2.8cm,
    minimum height=1.25cm,
    inner sep=5pt
  },
  theorem/.style={
    draw=secfiveline,
    rounded corners=0pt,
    fill=secfivetheorem,
    line width=.8pt,
    align=center,
    font=\footnotesize\bfseries,
    text width=3.5cm,
    minimum height=1.0cm,
    inner sep=6pt
  },
  arrow/.style={->,>=stealth,line width=.65pt,draw=secfiveline},
  line/.style={line width=.65pt,draw=secfiveline}
]

\node[box,text width=4.6cm] (unbonused) at (-3.75,6.1)
  {Unbonused loss and gain\\(Lemma~\ref{lem:unbonused-local})};
\node[box,text width=4cm] (full) at (3.75,6.1)
  {Total loss and gain\\(Lemma~\ref{lem:bonus-table})};

\node[resultbox,text width=3cm] (welldefined) at (-7.5,3.0)
  {Well-definedness\\(Lemma~\ref{lem:rule2-well-defined})};
\node[casebox,text width=3cm] (nochange) at (-3.75,3.0)
  {No bonus change};
\node[casebox,text width=3cm] (nocw) at (0,3.0)
  {Bonus change,\\no Condorcet winner\\(Lemma~\ref{lem:no-cw-bonus-cases})};
\node[casebox,text width=3cm] (newcw) at (3.75,3.0)
  {Bonus change,\\new Condorcet winner\\(Lemma~\ref{lem:condorcet-creation})};
\node[resultbox,text width=3cm] (monotonicity) at (7.5,3.0)
  {Monotonicity\\(Lemma~\ref{lem:bonus-monotonicity})};

\node[resultbox] (nm) at (0,.45) {$2$-NM$_2$};
\node[theorem] (theorem) at (0,-2.0)
  {Theorem~\ref{the:2NM}};

\draw[arrow] (unbonused) -- (full);
\draw[arrow] (unbonused) -- (nochange);

\coordinate (split) at (3.75,4.65);
\draw[line] (full.south) -- (split);
\draw[line] (0,4.65) -- (7.5,4.65);
\draw[arrow] (0,4.65) -- (nocw.north);
\draw[arrow] (3.75,4.65) -- (newcw.north);
\draw[arrow] (7.5,4.65) -- (monotonicity.north);

\coordinate (casebus) at (0,1.65);
\draw[line] (nochange.south) -- (-3.75,1.65);
\draw[line] (nocw.south) -- (0,1.65);
\draw[line] (newcw.south) -- (3.75,1.65);
\draw[line] (-3.75,1.65) -- (3.75,1.65);
\draw[arrow] (casebus) -- (nm.north);

\coordinate (theorembus) at (0,-.85);
\draw[line] (welldefined.south) -- (-7.5,-.85);
\draw[line] (monotonicity.south) -- (7.5,-.85);
\draw[line] (-7.5,-.85) -- (7.5,-.85);
\draw[line] (nm.south) -- (theorembus);
\draw[arrow] (theorembus) -- (theorem.north);

\end{tikzpicture}%
}
\endgroup
\caption{A roadmap of the proof of Theorem~\ref{the:2NM}. Three light green boxes are core parts.}
\label{fig:roadmap}
\end{figure}

\medskip

For tournaments without a Condorcet winner, the normalization in Rule~\ref{rule:2NM} makes the values $r_i(T)$ sum to one. Step $3$ also covers all cases since there are at most three almost-Condorcet winners. It remains to show that every $r_i(T)$ is nonnegative, as established in Lemma~\ref{lem:rule2-well-defined}. The detailed proof is deferred to Appendix~\ref{app:welldefine}.

\begin{lemma}[]
\label{lem:rule2-well-defined}
Rule~\ref{rule:2NM} defines a probability distribution on every tournament.
\end{lemma}

To analyze manipulation, we first compare the loss and gain due to the unbonused win-strength scores. Fix \(\{i,j\}\)-adjacent tournaments \(T,T'\) on \(N\) such that \(i\to_T j\) and \(j\to_{T'}i\), and suppose that neither tournament has a Condorcet winner. Let \(u\) and \(v\) denote the numbers of losses of \(i\) and \(j\) to teams outside \(\{i,j\}\), respectively. Thus, we have $s_i(T)=3^{1-u}$ and $s_j(T)=3^{-v}$ with $u,v\ge1$. After the reversal, \(s_i(T')=\frac{s_i(T)}{3}\) and \(s_j(T')=3s_j(T)\).

Let the total unbonused win-strength score be $W^0(T)=\sum_{q\in N}t_k^0(T)$,
and, for \(q\in N\), apply the normalization in Rule~\ref{rule:2NM} to the unbonused scores to define
\[
    r_k^0(T)
    =
    \frac{t_k^0(T)}{M}
    +\frac1n\left(1-\frac{W^0(T)}{M}\right).
\]
Define \(W^0(T')\) and \(r_k^0(T')\) analogously. Recall that, for a monotone rule, the definition of $2$-NM$_\lambda$ requires \(j\)'s gain in winning probability to be at most \(\lambda+1\) times \(i\)'s loss. We establish monotonicity separately below. For \(\lambda=2\), our objective is
\begin{equation}
\label{eq:unbonused-gain-loss}
    r_j(T')-r_j(T)
    \le
    3\bigl(r_i(T)-r_i(T')\bigr).
\end{equation}
To study this comparison for the unbonused scores, define the scaled unbonused loss of \(i\) and gain of \(j\) by
\[
    L_0
    :=
    nM\bigl(r_i^0(T)-r_i^0(T')\bigr),
    \quad
    G_0
    :=
    nM\bigl(r_j^0(T')-r_j^0(T)\bigr).
\]
For the unbonused scores, this target takes the form \(G_0\le3L_0\), or equivalently \(3L_0-G_0\ge0\).

\begin{lemma}[]
\label{lem:unbonused-local}
The scaled unbonused loss and gain satisfy
\begin{equation}
\label{eq:base-slack}
    3L_0-G_0
    =
    8\left(
        \frac{n-1-u}{3^u}
        +
        \frac{v+1}{3^v}
    \right)
    >0.
\end{equation}
In particular, if every team's bonus is unchanged between \(T\) and \(T'\), then \eqref{eq:unbonused-gain-loss} holds.
\end{lemma}

\begin{proof}
The unbonused win-strength score of team \(i\) changes by
\[
    t_i^0(T')-t_i^0(T)
    =
    \frac32\left(\frac{s_i(T)}{3}-s_i(T)\right)-s_j(T)
    =
    -s_i(T)-s_j(T)
\]
and that of team \(j\) changes by
\[
    t_j^0(T')-t_j^0(T)
    =
    \frac32(3s_j(T)-s_j(T))+\frac{s_i(T)}{3}
    =
    3s_j(T)+\frac{s_i(T)}{3}.
\]
For each team outside \(\{i,j\}\), the contribution of \(i\) to its unbonused win-strength score decreases by \(2s_i(T)/3\) if it defeats \(i\), and the contribution of \(j\) increases by \(2s_j(T)\) if it defeats \(j\). Combining these changes with those of \(i\) and \(j\) gives
\[
    W^0(T')-W^0(T)
    =
    -\frac23(u+1)s_i(T)+2(v+1)s_j(T).
\]
To express \(L_0\) and \(G_0\) in terms of these score changes, use
\[
    nMr_k^0(T)=nt_k^0(T)+M-W^0(T),
\]
which yields
\begin{align}
\label{eq:L0}
    L_0
    &=
    \left(n-1+\frac13-\frac{2u}{3}\right)s_i(T)
    +
    (n-1+3+2v)s_j(T),\\
    G_0
    &=
    \left(\frac{n-1}{3}+1+\frac{2u}{3}\right)s_i(T)
    +
    (3(n-1)+1-2v)s_j(T).\notag
\end{align}
Subtracting \(G_0\) from \(3L_0\) gives
\[
    3L_0-G_0
    =
    \frac{8(n-1-u)}{3}\,s_i(T)+8(v+1)s_j(T)
    =
    8\left(
        \frac{n-1-u}{3^u}
        +
        \frac{v+1}{3^v}
    \right).
\]
Since \(u\le n-2\), both terms in the last expression are nonnegative and the first is strictly positive. When every team's bonus is unchanged, \(L_0\) and \(G_0\) are also the loss of \(i\) and gain of \(j\) under the full rule, scaled by \(nM\). Thus \eqref{eq:base-slack} implies \eqref{eq:unbonused-gain-loss}.
\end{proof}

We continue to assume that neither \(T\) nor \(T'\) has a Condorcet winner. Lemma~\ref{lem:unbonused-local} establishes the desired gain--loss inequality when every team's bonus is unchanged. To account for changes in the bonuses, define the scaled loss of \(i\) and gain of \(j\) under the full rule by
\[
    L=nM\bigl(r_i(T)-r_i(T')\bigr),
    \quad
    G=nM\bigl(r_j(T')-r_j(T)\bigr).
\]
The target \eqref{eq:unbonused-gain-loss} is now \(G\le3L\), or equivalently \(3L-G\ge0\). Let
\[
    \delta_i=b_i(T')-b_i(T),
    \quad
    \delta_j=b_j(T')-b_j(T)
\]
denote the changes in the bonuses of \(i\) and \(j\), respectively, and let $\Delta=\sum_{q\in N}b_q(T')-\sum_{q\in N}b_q(T)$ denote the change in the total bonus.

\begin{lemma}[]
\label{lem:bonus-table}
The full scaled loss and gain satisfy
\begin{equation}
\label{eq:master-bonus}
    L=L_0-n\delta_i+\Delta,
    \quad
    G=G_0+n\delta_j-\Delta,
\end{equation}
and therefore
\begin{equation}
\label{eq:3LminusG}
    3L-G
    =
    3L_0-G_0
    +
    4\Delta
    -
    3n\delta_i
    -
    n\delta_j.
\end{equation}
Moreover, Table~\ref{tab:bonus-transitions} lists all possible changes in the bonuses.
\end{lemma}

\begin{table*}[!t]
\centering
\setlength{\belowcaptionskip}{4pt}
\small
\setlength{\tabcolsep}{5.5pt}
\renewcommand{\arraystretch}{1.12}
\begin{tabular}{cclccc}
\toprule
Case & \(|H(T)|\to|H(T')|\) & Role change & \(\delta_i\) & \(\delta_j\) & \(\Delta\)\\
\midrule
(a) & fixed & \makecell[l]{almost-Condorcet winners\\and blocker roles unchanged}
    & \(0\) & \(0\) & \(0\)\\
(b) & \(0\to1\) & \(j\) enters \(H\)
    & \(0\) & \(0\) & \(B\)\\
(c) & \(1\to0\) & \(i\) leaves \(H\)
    & \(0\) & \(0\) & \(-B\)\\
(d) & \(1\to1\) & \(i\) is replaced by \(j\)
    & \(0\) & \(0\) & \(0\)\\
(e\(_1\)) & \(1\to2\) & \(j:B\to P\)
    & \(0\) & \(P-B=-11/36\) & \(10/9\)\\
(e\(_2\)) & \(1\to2\) & old blocker \(i:B\to Q\)
    & \(Q-B=-19/18\) & \(0\) & \(10/9\)\\
(e\(_3\)) & \(1\to2\) & old blocker \(\ell\notin\{i,j\}\)
    & \(0\) & \(0\) & \(10/9\)\\
(f\(_1\)) & \(2\to1\) & \(i:P\to B\)
    & \(B-P=11/36\) & \(0\) & \(-10/9\)\\
(f\(_2\)) & \(2\to1\) & external blocker \(j:Q\to B\)
    & \(0\) & \(B-Q=19/18\) & \(-10/9\)\\
(f\(_3\)) & \(2\to1\) & external blocker \(\ell\notin\{i,j\}\)
    & \(0\) & \(0\) & \(-10/9\)\\
(g\(_1\)) & \(2\to2\) & \(i:P\to Q\)
    & \(Q-P=-3/4\) & \(0\) & \(0\)\\
(g\(_2\)) & \(2\to2\) & \(j:Q\to P\)
    & \(0\) & \(P-Q=3/4\) & \(0\)\\
(g\(_3\)) & \(2\to2\) & \makecell[l]{neither \(i\)'s nor \(j\)'s\\bonus changes}
    & \(0\) & \(0\) & \(0\)\\
(h) & \(2\to3\) & \(j:Q\to R\)
    & \(0\) & \(R-Q=-1/4\) & \(-1/12\)\\
(i) & \(3\to2\) & \(i:R\to Q\)
    & \(Q-R=1/4\) & \(0\) & \(1/12\)\\
(j) & \(3\to3\) & exchange: \(i\) is replaced by \(j\)
    & \multicolumn{3}{c}{impossible}\\
\bottomrule
\end{tabular}
\caption{All possible bonus changes under the reversal
\(i\to_T j\) to \(j\to_{T'}i\), assuming neither tournament has a
Condorcet winner. An annotation such as \(j:B\to P\) records \(j\)'s bonus before and after the reversal.}
\label{tab:bonus-transitions}
\end{table*}

\begin{proof}
Let \(W(T)=\sum_{q\in N}t_q(T)\) denote the total score, with \(W(T')\) defined analogously. The normalization in Rule~\ref{rule:2NM} gives
\[
    nMr_q(T)=nt_q(T)+M-W(T).
\]
The bonus changes therefore contribute \(-n\delta_i+\Delta\) to the scaled loss of \(i\) and \(n\delta_j-\Delta\) to the scaled gain of \(j\). This yields \eqref{eq:master-bonus} and \eqref{eq:3LminusG}.

To establish that the table is exhaustive, consider how the reversal changes the almost-Condorcet winners. Only the win counts of \(i\) and \(j\) change: \(i\)'s decreases by one, and \(j\)'s increases by one. Thus only \(i\) can leave the set of almost-Condorcet winners, only \(j\) can enter it, and every other member remains. Since \(|H(T)|,|H(T')|\le3\), only the cardinality changes shown in Table~\ref{tab:bonus-transitions} are possible.

The numerical entries follow only from \(|H(T)|\) and \(|H(T')|\), while the almost-Condorcet structure determines the role changes. For example, in case (f\(_1\)), \(H(T)=\{i,h\}\) with \(i\to_T h\). Team \(i\) is the internal blocker in \(T\); after leaving the set of almost-Condorcet winners, it becomes the unique blocker of \(h\), so its bonus changes from \(P\) to \(B\). In case (f\(_2\)), the almost-Condorcet winner that remains after the reversal defeats \(i\), and its external blocker is \(j\). The bonus of \(j\) therefore changes from \(Q\) to \(B\). Figures~\ref{fig:sec5-transition-f1} and~\ref{fig:sec5-transition-f2} illustrate these two changes in blocker roles as examples: each diagram shows \(T\); the thick arrow marks \(i\to_T j\). Gold shading and double outlines identify the almost-Condorcet winners in \(T\). The annotation tables record win counts in \(T\) and bonus transitions from \(T\) to \(T'\). For diagrams of the remaining configurations in cases (b)--(i) of Table~\ref{tab:bonus-transitions}, see Appendix~\ref{app:bonus-transitions}. The remaining rows follow from the same structural description or by reversing these arguments.

\suppressfloats[t]
\begin{figure}[t]
\centering
\captionsetup{font=footnotesize,skip=4pt}
\begin{minipage}[t]{.48\linewidth}
\centering
\resizebox{\dimexpr\linewidth-6pt\relax}{!}{%
  \begingroup
\definecolor{bonusblue}{HTML}{24576E}
\definecolor{acgold}{HTML}{FFE8A3}
\begin{tikzpicture}[
  team/.style={circle,draw=black,fill=white,line width=.55pt,
    minimum size=5.2mm,inner sep=0pt,font=\small},
  ac team/.style={team,fill=acgold,double,double distance=.7pt,line width=.45pt},
  match/.style={-{Stealth[length=1.9mm,width=1.3mm]},line width=.6pt},
  reversed match/.style={match,line width=1.1pt,draw=bonusblue},
  every node/.style={font=\small}
]
\path[use as bounding box] (0,0) rectangle (7.4,4.45);

\node[anchor=west,inner sep=0pt,font=\bfseries\small] at (.1,4.25) {Case (f$_1$)};
\node[anchor=east,inner sep=0pt,font=\footnotesize] at (7.3,4.25) {$|H(T)|=2,\quad |H(T')|=1$};
\node[ac team] (i) at (2.700,2.070) {$\boldsymbol{i}$};
\node[team] (j) at (6.200,2.070) {$j$};
\node[ac team] (h) at (4.450,3.435) {$\boldsymbol{h}$};
\node[team] (ell) at (0.900,3.045) {$\ell$};
\draw[match] (ell) -- (i);
\draw[match] (i) -- (h);
\draw[match] (h) -- (j);
\draw[reversed match] (i) -- (j);
\node[anchor=north,inner sep=0pt,font=\fontsize{8.5}{10}\selectfont] at (3.7,1.5) {
\setlength{\tabcolsep}{3pt}
\renewcommand{\arraystretch}{1.0}
\begin{tabular}{@{}p{1.70cm}>{\columncolor{acgold}\centering\arraybackslash}p{1.08cm}>{\centering\arraybackslash}p{1.08cm}>{\columncolor{acgold}\centering\arraybackslash}p{1.08cm}>{\centering\arraybackslash}p{1.08cm}@{}}
\toprule
\textbf{Team} & $\boldsymbol{i}$ & $j$ & $\boldsymbol{h}$ & $\ell$ \\
\midrule
\textbf{Wins in $T$} & $n-2$ & \textemdash & $n-2$ & \textemdash \\
\textbf{Bonus} & $\color{bonusblue}P\to B$ & \textemdash & \textemdash & $\color{bonusblue}Q\to0$ \\
\bottomrule
\end{tabular}
};
\end{tikzpicture}
\endgroup%
}
\caption{Case (f\(_1\)): \(i\)'s bonus changes from \(P\) to \(B\).}
\label{fig:sec5-transition-f1}
\end{minipage}\hfill
\begin{minipage}[t]{.48\linewidth}
\centering
\resizebox{\dimexpr\linewidth-6pt\relax}{!}{%
  \begingroup
\definecolor{bonusblue}{HTML}{24576E}
\definecolor{acgold}{HTML}{FFE8A3}
\begin{tikzpicture}[
  team/.style={circle,draw=black,fill=white,line width=.55pt,
    minimum size=5.2mm,inner sep=0pt,font=\small},
  ac team/.style={team,fill=acgold,double,double distance=.7pt,line width=.45pt},
  match/.style={-{Stealth[length=1.9mm,width=1.3mm]},line width=.6pt},
  reversed match/.style={match,line width=1.1pt,draw=bonusblue},
  every node/.style={font=\small}
]
\path[use as bounding box] (0,0) rectangle (7.4,4.45);

\node[anchor=west,inner sep=0pt,font=\bfseries\small] at (.1,4.25) {Case (f$_2$)};
\node[anchor=east,inner sep=0pt,font=\footnotesize] at (7.3,4.25) {$|H(T)|=2,\quad |H(T')|=1$};
\node[ac team] (i) at (1.500,2.070) {$\boldsymbol{i}$};
\node[team] (j) at (5.900,2.070) {$j$};
\node[ac team] (h) at (3.700,3.435) {$\boldsymbol{h}$};
\draw[match] (h) -- (i);
\draw[reversed match] (i) -- (j);
\draw[match] (j) -- (h);
\node[anchor=north,inner sep=0pt,font=\fontsize{8.5}{10}\selectfont] at (3.7,1.5) {
\setlength{\tabcolsep}{3pt}
\renewcommand{\arraystretch}{1.0}
\begin{tabular}{@{}p{1.70cm}>{\columncolor{acgold}\centering\arraybackslash}p{1.51cm}>{\centering\arraybackslash}p{1.51cm}>{\columncolor{acgold}[3pt][0pt]\centering\arraybackslash}p{1.51cm}@{}}
\toprule
\textbf{Team} & $\boldsymbol{i}$ & $j$ & $\boldsymbol{h}$ \\
\midrule
\textbf{Wins in $T$} & $n-2$ & \textemdash & $n-2$ \\
\textbf{Bonus} & \textemdash & $\color{bonusblue}Q\to B$ & $\color{bonusblue}P\to0$ \\
\bottomrule
\end{tabular}
};
\end{tikzpicture}
\endgroup%
}
\caption{Case (f\(_2\)): \(j\)'s bonus changes from \(Q\) to \(B\).}
\label{fig:sec5-transition-f2}
\end{minipage}
\par\smallskip
\begin{minipage}{\linewidth}

\end{minipage}
\end{figure}

Finally, the exchange in case (j) is impossible. The three almost-Condorcet winners in \(T\) form a directed cycle and defeat every team outside \(H(T)\), including \(j\). The two teams that remain almost-Condorcet winners after the reversal therefore both defeat \(j\). Thus, $j$ cannot be an almost-Condorcet winner.
\end{proof}

Recall that \(u\) and \(v\) are the losses of \(i\) and \(j\) to teams outside \(\{i,j\}\), respectively. The analysis for \(G\le3L\) uses the following three properties on \(n,u,v\) in small tournaments, and we defer the proofs to Appendix~\ref{app:property}.

\begin{lemma}[]
\label{lem:small-feasibility}
Assume that neither \(T\) nor \(T'\) has a Condorcet winner.
\begin{enumerate}
\item If \(|H(T)|=1\) and \(|H(T')|=0\), then \(n=4\) is impossible. Moreover, if \(n=5\), then \(v=3\) is impossible.
\item If \(|H(T)|=2\) and \(|H(T')|=1\), then \(n\ge5\).
\item In case \emph{(f\(_2\))}, one has \(n\ge v+3\).
\end{enumerate}
\end{lemma}

Cases in which every team's bonus is unchanged are already covered by Lemma~\ref{lem:unbonused-local}. When the bonuses change, we now use the feasibility constraints to prove \(G\le3L\) for all transitions in which neither tournament has a Condorcet winner.  The detailed analysis is deferred to Appendix~\ref{app:proof-no-cw-bonus-cases}.

\begin{lemma}[]
\label{lem:no-cw-bonus-cases}
If neither \(T\) nor \(T'\) has a Condorcet winner, then every case in Table~\ref{tab:bonus-transitions} satisfies $G\le3L$.

\end{lemma}

We next consider a reversal that makes \(j\) a Condorcet winner, so that the rule selects \(j\) with probability one. The bonuses ensure that \(i\)'s winning probability in \(T\) is sufficiently large for the gain--loss inequality to hold. The detailed case analysis is deferred to Appendix~\ref{app:proof-condorcet-creation}.

\begin{lemma}[]
\label{lem:condorcet-creation}
Suppose \(T\) has no Condorcet winner and \(j\) becomes a Condorcet winner in \(T'\). Then $r_j(T')-r_j(T)
    \le
    3\bigl(r_i(T)-r_i(T')\bigr).$
\end{lemma}

Monotonicity requires \(r_i(T)\ge r_i(T')\). If a Condorcet winner is either $i$, $j$, or any other team, this trivially holds. When neither tournament has a Condorcet winner, we compare the scaled loss \(L\) with the scaled unbonused loss \(L_0>0\). The inequality \(L\ge L_0\) already gives monotonicity, so only transitions in which the bonuses can decrease \(L\) below \(L_0\) require further analysis. The detailed analysis is deferred to Appendix~\ref{app:proof-bonus-monotonicity}.

\begin{lemma}[]
\label{lem:bonus-monotonicity}
If neither \(T\) nor \(T'\) has a Condorcet winner, then every transition in Table~\ref{tab:bonus-transitions} satisfies
$r_i(T)\ge r_i(T').$
\end{lemma}

We now combine the preceding lemmas to prove the main result.

\begin{proof}[Proof of Theorem~\ref{the:2NM}]
The rule is well-defined by Lemma~\ref{lem:rule2-well-defined} and Condorcet consistent by construction. Monotonicity follows from Lemma~\ref{lem:bonus-monotonicity} and the preceding Condorcet observations.

Fix \(\{i,j\}\)-adjacent tournaments \(T,T'\) with \(i\to_T j\) and \(j\to_{T'}i\). To prove \(2\)-NM\(_2\), it suffices to show $r_j(T')-r_j(T)\le3\bigl(r_i(T)-r_i(T')\bigr)$ by equation \ref{eq:gain-loss}.
If neither tournament has a Condorcet winner, Lemma~\ref{lem:unbonused-local} applies when every bonus is unchanged, and Lemma~\ref{lem:no-cw-bonus-cases} otherwise.
Lemma~\ref{lem:condorcet-creation} covers the case where \(j\) becomes a Condorcet winner. If \(i\) is a Condorcet winner in \(T\), probability conservation gives the stronger bound with factor \(1\). A Condorcet winner outside \(\{i,j\}\) remains a Condorcet winner, so \(i\)'s and \(j\)'s winning probabilities remain zero. Thus the rule is \(2\)-NM\(_2\).

For tightness, let \(\widehat{T}_m\) be a tournament on an odd number \(m\) of teams, each with \(\frac{m-1}{2}\) wins. Choose \(i\in\widehat{T}_m\) and add a team \(j\) with \(i\to j\) and \(j\to v\) for every \(v\in\widehat{T}_m\setminus\{i\}\).
The resulting complete tournament \(T_m\) has a unique almost-Condorcet winner \(j\), whose unique blocker is \(i\). The teams in \(\widehat{T}_m\) have scores \(s_i(T_m)=3^{-\frac{m-3}{2}}\) and \(s_v(T_m)=3^{-\frac{m-1}{2}}\) for \(v\in\widehat{T}_m\setminus\{i\}\). These scores sum to \((m+2)3^{-\frac{m-1}{2}}\longrightarrow0\) as \(m\to\infty\). Hence Rule~\ref{rule:2NM} gives \(t_i(T_m)\longrightarrow1+B\) and \(t_j(T_m)\longrightarrow\frac32\), while each team's uniform share of the remaining probability mass tends to zero.

In \(T_m'\), obtained by reversing \(i\to j\), team \(j\) is the Condorcet winner. Hence
\[
    \frac{r_j(T_m')-r_j(T_m)}
         {r_i(T_m)-r_i(T_m')}
    \longrightarrow
    \frac{M-\frac32}{1+B}
    =
    3.
\]
Therefore Rule~\ref{rule:2NM} is not \(2\)-NM\(_\lambda\) for any \(\lambda<2\).
\end{proof}

\section{Conclusion}

In this paper, we first establish a hierarchy showing $k$-NM$_\lambda$ implies $k$-MNM-$(1+\lambda)$ and hence $k$-SNM-$\lambda/(1+\lambda)$ for any coalition size $k$, while both converse implications fail. Then we consider pairwise non-manipulability. We first close the pairwise $2$-MNM-$\delta$ question by proving the Randomized Death Match rule is $2$-MNM-$3/2$, which is the first rule known to match the universal lower bound. Consequently, we design a novel rule satisfying $2$-non-manipulability for $\lambda=2$, substantially improving the best known upper bound of $\lambda=11$. 

Our results uncover several interesting questions. First, the gap between our upper bound $\lambda=2$ and the conjectured lower bound $\lambda=1$ remains open. Second, it is natural to research these non-manipulation notions on larger coalitions. Finally, one may extend the framework beyond round-robin tournaments and complete graphs. 
For example, tournaments where not every pair of teams plays are more common in practice.
Such settings would require suitably adapted definitions of Condorcet consistency.

\section*{AI Disclosure}
We used GPT 5.6 Sol to assist with the parameter calculation, figure visualization, and typo checking. The tool materially affected Section \ref{sec:2NM}. The authors verified the correctness and originality of all content including references.

\newpage

\bibliographystyle{plain}
\bibliography{ref}

\newpage

\appendix
\clearpage
\section{Illustrations of the Bonus Transitions}
\label{app:bonus-transitions}

The figures below illustrate cases (b)--(i) of Table~\ref{tab:bonus-transitions}, with each subcase shown separately. As in that table, $T$ and $T'$ are $\{i,j\}$-adjacent tournaments with $i\to_T j$ and $j\to_{T'}i$, and neither tournament has a Condorcet winner. Each diagram shows the indicated teams and match outcomes in $T$; the thicker arrow marks $i\to_T j$. Gold shading and double outlines identify the almost-Condorcet winners in $T$. The table below each diagram records the win count $d_u(T)$ and the bonus transition $b_u(T)\to b_u(T')$ for each displayed team $u$, with the columns for almost-Condorcet winners shaded as well. A dash denotes an omitted annotation. The labels $h$ and $\ell$ identify auxiliary teams within each figure; match outcomes that are not drawn are left unspecified.

\begingroup
\captionsetup{font=footnotesize,skip=6pt}
\vspace{18pt}

\noindent\begin{minipage}{\linewidth}
  \centering
  \begin{minipage}[t]{.48\linewidth}
  \captionsetup{type=figure}
  \centering
  \resizebox{.82\linewidth}{!}{%
    \begingroup
\definecolor{bonusblue}{HTML}{24576E}
\definecolor{acgold}{HTML}{FFE8A3}
\begin{tikzpicture}[
  team/.style={circle,draw=black,fill=white,line width=.55pt,
    minimum size=5.2mm,inner sep=0pt,font=\small},
  ac team/.style={team,fill=acgold,double,double distance=.7pt,line width=.45pt},
  match/.style={-{Stealth[length=1.9mm,width=1.3mm]},line width=.6pt},
  reversed match/.style={match,line width=1.1pt,draw=bonusblue},
  every node/.style={font=\small}
]
\path[use as bounding box] (0,0) rectangle (7.4,4.45);

\node[anchor=west,inner sep=0pt,font=\bfseries\small] at (.1,4.25) {Case (b)};
\node[anchor=east,inner sep=0pt,font=\footnotesize] at (7.3,4.25) {$|H(T)|=0,\quad |H(T')|=1$};
\node[team] (i) at (1.200,2.720) {$i$};
\node[team] (j) at (3.700,2.720) {$j$};
\node[team] (h) at (6.200,2.720) {$h$};
\draw[reversed match] (i) -- (j);
\draw[match] (h) -- (j);
\node[anchor=north,inner sep=0pt,font=\fontsize{8.5}{10}\selectfont] at (3.7,1.5) {
\setlength{\tabcolsep}{3pt}
\renewcommand{\arraystretch}{1.0}
\begin{tabular}{@{}p{1.70cm}>{\centering\arraybackslash}p{1.51cm}>{\centering\arraybackslash}p{1.51cm}>{\centering\arraybackslash}p{1.51cm}@{}}
\toprule
\textbf{Team} & $i$ & $j$ & $h$ \\
\midrule
\textbf{Wins in $T$} & \textemdash & $n-3$ & \textemdash \\
\textbf{Bonus} & \textemdash & \textemdash & $\color{bonusblue}0\to B$ \\
\bottomrule
\end{tabular}
};
\end{tikzpicture}
\endgroup%
  }
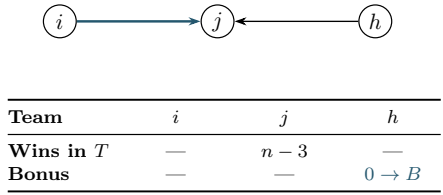
  \caption{Case (b), $0\to1$: $j$ enters $H$.}
  \label{fig:bonus-transition-b}
  \end{minipage}\hfill
  \begin{minipage}[t]{.48\linewidth}
  \captionsetup{type=figure}
  \centering
  \resizebox{.82\linewidth}{!}{%
    \begingroup
\definecolor{bonusblue}{HTML}{24576E}
\definecolor{acgold}{HTML}{FFE8A3}
\begin{tikzpicture}[
  team/.style={circle,draw=black,fill=white,line width=.55pt,
    minimum size=5.2mm,inner sep=0pt,font=\small},
  ac team/.style={team,fill=acgold,double,double distance=.7pt,line width=.45pt},
  match/.style={-{Stealth[length=1.9mm,width=1.3mm]},line width=.6pt},
  reversed match/.style={match,line width=1.1pt,draw=bonusblue},
  every node/.style={font=\small}
]
\path[use as bounding box] (0,0) rectangle (7.4,4.45);

\node[anchor=west,inner sep=0pt,font=\bfseries\small] at (.1,4.25) {Case (c)};
\node[anchor=east,inner sep=0pt,font=\footnotesize] at (7.3,4.25) {$|H(T)|=1,\quad |H(T')|=0$};
\node[ac team] (i) at (3.700,2.720) {$\boldsymbol{i}$};
\node[team] (j) at (6.200,2.720) {$j$};
\node[team] (h) at (1.200,2.720) {$h$};
\draw[match] (h) -- (i);
\draw[reversed match] (i) -- (j);
\node[anchor=north,inner sep=0pt,font=\fontsize{8.5}{10}\selectfont] at (3.7,1.5) {
\setlength{\tabcolsep}{3pt}
\renewcommand{\arraystretch}{1.0}
\begin{tabular}{@{}p{1.70cm}>{\columncolor{acgold}\centering\arraybackslash}p{1.51cm}>{\centering\arraybackslash}p{1.51cm}>{\centering\arraybackslash}p{1.51cm}@{}}
\toprule
\textbf{Team} & $\boldsymbol{i}$ & $j$ & $h$ \\
\midrule
\textbf{Wins in $T$} & $n-2$ & \textemdash & \textemdash \\
\textbf{Bonus} & \textemdash & \textemdash & $\color{bonusblue}B\to0$ \\
\bottomrule
\end{tabular}
};
\end{tikzpicture}
\endgroup%
  }
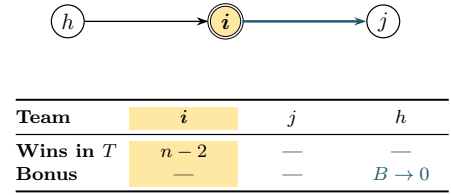
  \caption{Case (c), $1\to0$: $i$ leaves $H$.}
  \label{fig:bonus-transition-c}
  \end{minipage}
\end{minipage}\par
\vspace{28pt}

\noindent\begin{minipage}{\linewidth}
  \centering
  \begin{minipage}[t]{.48\linewidth}
  \captionsetup{type=figure}
  \centering
  \resizebox{.82\linewidth}{!}{%
    \begingroup
\definecolor{bonusblue}{HTML}{24576E}
\definecolor{acgold}{HTML}{FFE8A3}
\begin{tikzpicture}[
  team/.style={circle,draw=black,fill=white,line width=.55pt,
    minimum size=5.2mm,inner sep=0pt,font=\small},
  ac team/.style={team,fill=acgold,double,double distance=.7pt,line width=.45pt},
  match/.style={-{Stealth[length=1.9mm,width=1.3mm]},line width=.6pt},
  reversed match/.style={match,line width=1.1pt,draw=bonusblue},
  every node/.style={font=\small}
]
\path[use as bounding box] (0,0) rectangle (7.4,4.45);

\node[anchor=west,inner sep=0pt,font=\bfseries\small] at (.1,4.25) {Case (d)};
\node[anchor=east,inner sep=0pt,font=\footnotesize] at (7.3,4.25) {$|H(T)|=1,\quad |H(T')|=1$};
\node[ac team] (i) at (1.500,2.200) {$\boldsymbol{i}$};
\node[team] (j) at (5.900,2.200) {$j$};
\node[team] (h) at (1.500,3.370) {$h$};
\node[team] (ell) at (5.900,3.370) {$\ell$};
\draw[match] (h) -- (i);
\draw[reversed match] (i) -- (j);
\draw[match] (ell) -- (j);
\node[anchor=north,inner sep=0pt,font=\fontsize{8.5}{10}\selectfont] at (3.7,1.5) {
\setlength{\tabcolsep}{3pt}
\renewcommand{\arraystretch}{1.0}
\begin{tabular}{@{}p{1.70cm}>{\columncolor{acgold}\centering\arraybackslash}p{1.08cm}>{\centering\arraybackslash}p{1.08cm}>{\centering\arraybackslash}p{1.08cm}>{\centering\arraybackslash}p{1.08cm}@{}}
\toprule
\textbf{Team} & $\boldsymbol{i}$ & $j$ & $h$ & $\ell$ \\
\midrule
\textbf{Wins in $T$} & $n-2$ & $n-3$ & \textemdash & \textemdash \\
\textbf{Bonus} & \textemdash & \textemdash & $\color{bonusblue}B\to0$ & $\color{bonusblue}0\to B$ \\
\bottomrule
\end{tabular}
};
\end{tikzpicture}
\endgroup%
  }
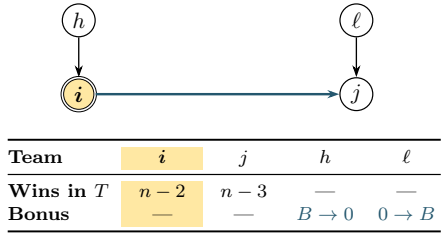
  \caption{Case (d), $1\to1$: $j$ replaces $i$ in $H$.}
  \label{fig:bonus-transition-d}
  \end{minipage}\hfill
  \begin{minipage}[t]{.48\linewidth}
  \captionsetup{type=figure}
  \centering
  \resizebox{.82\linewidth}{!}{%
    \begingroup
\definecolor{bonusblue}{HTML}{24576E}
\definecolor{acgold}{HTML}{FFE8A3}
\begin{tikzpicture}[
  team/.style={circle,draw=black,fill=white,line width=.55pt,
    minimum size=5.2mm,inner sep=0pt,font=\small},
  ac team/.style={team,fill=acgold,double,double distance=.7pt,line width=.45pt},
  match/.style={-{Stealth[length=1.9mm,width=1.3mm]},line width=.6pt},
  reversed match/.style={match,line width=1.1pt,draw=bonusblue},
  every node/.style={font=\small}
]
\path[use as bounding box] (0,0) rectangle (7.4,4.45);

\node[anchor=west,inner sep=0pt,font=\bfseries\small] at (.1,4.25) {Case (e$_1$)};
\node[anchor=east,inner sep=0pt,font=\footnotesize] at (7.3,4.25) {$|H(T)|=1,\quad |H(T')|=2$};
\node[team] (i) at (1.200,2.525) {$i$};
\node[team] (j) at (3.700,2.525) {$j$};
\node[ac team] (h) at (6.200,2.525) {$\boldsymbol{h}$};
\node[team] (ell) at (3.700,3.565) {$\ell$};
\draw[reversed match] (i) -- (j);
\draw[match] (ell) -- (j);
\draw[match] (j) -- (h);
\node[anchor=north,inner sep=0pt,font=\fontsize{8.5}{10}\selectfont] at (3.7,1.5) {
\setlength{\tabcolsep}{3pt}
\renewcommand{\arraystretch}{1.0}
\begin{tabular}{@{}p{1.70cm}>{\centering\arraybackslash}p{1.08cm}>{\centering\arraybackslash}p{1.08cm}>{\columncolor{acgold}\centering\arraybackslash}p{1.08cm}>{\centering\arraybackslash}p{1.08cm}@{}}
\toprule
\textbf{Team} & $i$ & $j$ & $\boldsymbol{h}$ & $\ell$ \\
\midrule
\textbf{Wins in $T$} & $<n-2$ & $n-3$ & $n-2$ & \textemdash \\
\textbf{Bonus} & \textemdash & $\color{bonusblue}B\to P$ & \textemdash & $\color{bonusblue}0\to Q$ \\
\bottomrule
\end{tabular}
};
\end{tikzpicture}
\endgroup%
  }
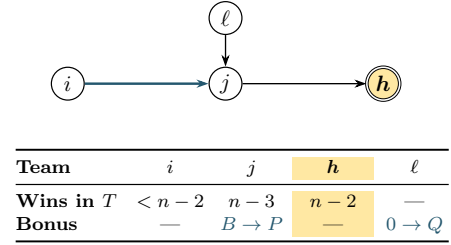
  \caption{Case (e$_1$), $1\to2$: the bonus of $j$ changes from $B$ to $P$.}
  \label{fig:bonus-transition-e1}
  \end{minipage}
\end{minipage}\par
\vspace{28pt}

\noindent\begin{minipage}{\linewidth}
  \centering
  \begin{minipage}[t]{.48\linewidth}
  \captionsetup{type=figure}
  \centering
  \resizebox{.82\linewidth}{!}{%
    \begingroup
\definecolor{bonusblue}{HTML}{24576E}
\definecolor{acgold}{HTML}{FFE8A3}
\begin{tikzpicture}[
  team/.style={circle,draw=black,fill=white,line width=.55pt,
    minimum size=5.2mm,inner sep=0pt,font=\small},
  ac team/.style={team,fill=acgold,double,double distance=.7pt,line width=.45pt},
  match/.style={-{Stealth[length=1.9mm,width=1.3mm]},line width=.6pt},
  reversed match/.style={match,line width=1.1pt,draw=bonusblue},
  every node/.style={font=\small}
]
\path[use as bounding box] (0,0) rectangle (7.4,4.45);

\node[anchor=west,inner sep=0pt,font=\bfseries\small] at (.1,4.25) {Case (e$_2$)};
\node[anchor=east,inner sep=0pt,font=\footnotesize] at (7.3,4.25) {$|H(T)|=1,\quad |H(T')|=2$};
\node[team] (i) at (1.500,3.305) {$i$};
\node[team] (j) at (5.900,3.305) {$j$};
\node[ac team] (h) at (3.700,2.070) {$\boldsymbol{h}$};
\draw[reversed match] (i) -- (j);
\draw[match] (i) -- (h);
\draw[match] (h) -- (j);
\node[anchor=north,inner sep=0pt,font=\fontsize{8.5}{10}\selectfont] at (3.7,1.5) {
\setlength{\tabcolsep}{3pt}
\renewcommand{\arraystretch}{1.0}
\begin{tabular}{@{}p{1.70cm}>{\centering\arraybackslash}p{1.51cm}>{\centering\arraybackslash}p{1.51cm}>{\columncolor{acgold}[3pt][0pt]\centering\arraybackslash}p{1.51cm}@{}}
\toprule
\textbf{Team} & $i$ & $j$ & $\boldsymbol{h}$ \\
\midrule
\textbf{Wins in $T$} & $<n-2$ & $n-3$ & $n-2$ \\
\textbf{Bonus} & $\color{bonusblue}B\to Q$ & \textemdash & $\color{bonusblue}0\to P$ \\
\bottomrule
\end{tabular}
};
\end{tikzpicture}
\endgroup%
  }
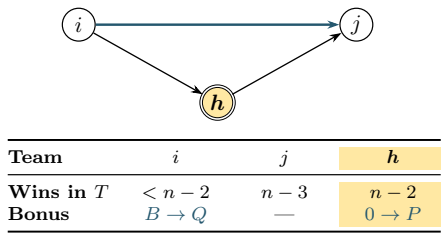
  \caption{Case (e$_2$), $1\to2$: $i$ is the old blocker, and its bonus changes from $B$ to $Q$.}
  \label{fig:bonus-transition-e2}
  \end{minipage}\hfill
  \begin{minipage}[t]{.48\linewidth}
  \captionsetup{type=figure}
  \centering
  \resizebox{.82\linewidth}{!}{%
    \begingroup
\definecolor{bonusblue}{HTML}{24576E}
\definecolor{acgold}{HTML}{FFE8A3}
\begin{tikzpicture}[
  team/.style={circle,draw=black,fill=white,line width=.55pt,
    minimum size=5.2mm,inner sep=0pt,font=\small},
  ac team/.style={team,fill=acgold,double,double distance=.7pt,line width=.45pt},
  match/.style={-{Stealth[length=1.9mm,width=1.3mm]},line width=.6pt},
  reversed match/.style={match,line width=1.1pt,draw=bonusblue},
  every node/.style={font=\small}
]
\path[use as bounding box] (0,0) rectangle (7.4,4.45);

\node[anchor=west,inner sep=0pt,font=\bfseries\small] at (.1,4.25) {Case (e$_3$)};
\node[anchor=east,inner sep=0pt,font=\footnotesize] at (7.3,4.25) {$|H(T)|=1,\quad |H(T')|=2$};
\node[team] (i) at (1.500,2.070) {$i$};
\node[team] (j) at (5.900,2.070) {$j$};
\node[ac team] (h) at (5.900,3.435) {$\boldsymbol{h}$};
\node[team] (ell) at (1.500,3.435) {$\ell$};
\draw[reversed match] (i) -- (j);
\draw[match] (h) -- (j);
\draw[match] (ell) -- (h);
\node[anchor=north,inner sep=0pt,font=\fontsize{8.5}{10}\selectfont] at (3.7,1.5) {
\setlength{\tabcolsep}{3pt}
\renewcommand{\arraystretch}{1.0}
\begin{tabular}{@{}p{1.70cm}>{\centering\arraybackslash}p{1.08cm}>{\centering\arraybackslash}p{1.08cm}>{\columncolor{acgold}\centering\arraybackslash}p{1.08cm}>{\centering\arraybackslash}p{1.08cm}@{}}
\toprule
\textbf{Team} & $i$ & $j$ & $\boldsymbol{h}$ & $\ell$ \\
\midrule
\textbf{Wins in $T$} & \textemdash & $n-3$ & $n-2$ & \textemdash \\
\textbf{Bonus} & \textemdash & \textemdash & $\color{bonusblue}0\to P$ & $\color{bonusblue}B\to Q$ \\
\bottomrule
\end{tabular}
};
\end{tikzpicture}
\endgroup%
  }
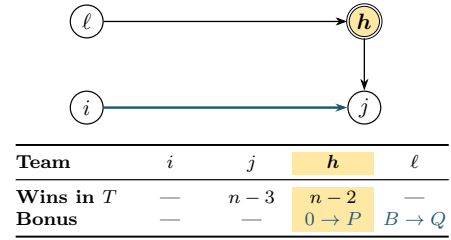
  \caption{Case (e$_3$), $1\to2$: the old blocker $\ell$ lies outside $\{i,j\}$.}
  \label{fig:bonus-transition-e3}
  \end{minipage}
\end{minipage}\par
\clearpage

\noindent\begin{minipage}{\linewidth}
  \centering
  \begin{minipage}[t]{.48\linewidth}
  \captionsetup{type=figure}
  \centering
  \resizebox{.82\linewidth}{!}{%
    \begingroup
\definecolor{bonusblue}{HTML}{24576E}
\definecolor{acgold}{HTML}{FFE8A3}
\begin{tikzpicture}[
  team/.style={circle,draw=black,fill=white,line width=.55pt,
    minimum size=5.2mm,inner sep=0pt,font=\small},
  ac team/.style={team,fill=acgold,double,double distance=.7pt,line width=.45pt},
  match/.style={-{Stealth[length=1.9mm,width=1.3mm]},line width=.6pt},
  reversed match/.style={match,line width=1.1pt,draw=bonusblue},
  every node/.style={font=\small}
]
\path[use as bounding box] (0,0) rectangle (7.4,4.45);

\node[anchor=west,inner sep=0pt,font=\bfseries\small] at (.1,4.25) {Case (f$_1$)};
\node[anchor=east,inner sep=0pt,font=\footnotesize] at (7.3,4.25) {$|H(T)|=2,\quad |H(T')|=1$};
\node[ac team] (i) at (2.700,2.070) {$\boldsymbol{i}$};
\node[team] (j) at (6.200,2.070) {$j$};
\node[ac team] (h) at (4.450,3.435) {$\boldsymbol{h}$};
\node[team] (ell) at (0.900,3.045) {$\ell$};
\draw[match] (ell) -- (i);
\draw[match] (i) -- (h);
\draw[match] (h) -- (j);
\draw[reversed match] (i) -- (j);
\node[anchor=north,inner sep=0pt,font=\fontsize{8.5}{10}\selectfont] at (3.7,1.5) {
\setlength{\tabcolsep}{3pt}
\renewcommand{\arraystretch}{1.0}
\begin{tabular}{@{}p{1.70cm}>{\columncolor{acgold}\centering\arraybackslash}p{1.08cm}>{\centering\arraybackslash}p{1.08cm}>{\columncolor{acgold}\centering\arraybackslash}p{1.08cm}>{\centering\arraybackslash}p{1.08cm}@{}}
\toprule
\textbf{Team} & $\boldsymbol{i}$ & $j$ & $\boldsymbol{h}$ & $\ell$ \\
\midrule
\textbf{Wins in $T$} & $n-2$ & \textemdash & $n-2$ & \textemdash \\
\textbf{Bonus} & $\color{bonusblue}P\to B$ & \textemdash & \textemdash & $\color{bonusblue}Q\to0$ \\
\bottomrule
\end{tabular}
};
\end{tikzpicture}
\endgroup%
  }
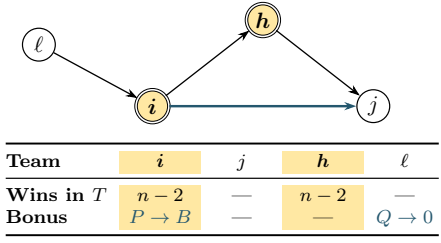
  \caption{Case (f$_1$), $2\to1$: the bonus of $i$ changes from $P$ to $B$.}
  \label{fig:bonus-transition-f1}
  \end{minipage}\hfill
  \begin{minipage}[t]{.48\linewidth}
  \captionsetup{type=figure}
  \centering
  \resizebox{.82\linewidth}{!}{%
    \begingroup
\definecolor{bonusblue}{HTML}{24576E}
\definecolor{acgold}{HTML}{FFE8A3}
\begin{tikzpicture}[
  team/.style={circle,draw=black,fill=white,line width=.55pt,
    minimum size=5.2mm,inner sep=0pt,font=\small},
  ac team/.style={team,fill=acgold,double,double distance=.7pt,line width=.45pt},
  match/.style={-{Stealth[length=1.9mm,width=1.3mm]},line width=.6pt},
  reversed match/.style={match,line width=1.1pt,draw=bonusblue},
  every node/.style={font=\small}
]
\path[use as bounding box] (0,0) rectangle (7.4,4.45);

\node[anchor=west,inner sep=0pt,font=\bfseries\small] at (.1,4.25) {Case (f$_2$)};
\node[anchor=east,inner sep=0pt,font=\footnotesize] at (7.3,4.25) {$|H(T)|=2,\quad |H(T')|=1$};
\node[ac team] (i) at (1.500,2.070) {$\boldsymbol{i}$};
\node[team] (j) at (5.900,2.070) {$j$};
\node[ac team] (h) at (3.700,3.435) {$\boldsymbol{h}$};
\draw[match] (h) -- (i);
\draw[reversed match] (i) -- (j);
\draw[match] (j) -- (h);
\node[anchor=north,inner sep=0pt,font=\fontsize{8.5}{10}\selectfont] at (3.7,1.5) {
\setlength{\tabcolsep}{3pt}
\renewcommand{\arraystretch}{1.0}
\begin{tabular}{@{}p{1.70cm}>{\columncolor{acgold}\centering\arraybackslash}p{1.51cm}>{\centering\arraybackslash}p{1.51cm}>{\columncolor{acgold}[3pt][0pt]\centering\arraybackslash}p{1.51cm}@{}}
\toprule
\textbf{Team} & $\boldsymbol{i}$ & $j$ & $\boldsymbol{h}$ \\
\midrule
\textbf{Wins in $T$} & $n-2$ & \textemdash & $n-2$ \\
\textbf{Bonus} & \textemdash & $\color{bonusblue}Q\to B$ & $\color{bonusblue}P\to0$ \\
\bottomrule
\end{tabular}
};
\end{tikzpicture}
\endgroup%
  }
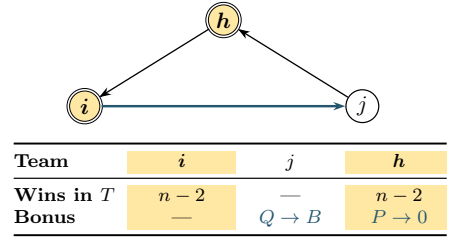
  \caption{Case (f$_2$), $2\to1$: $j$ is the external blocker, and its bonus changes from $Q$ to $B$.}
  \label{fig:bonus-transition-f2}
  \end{minipage}
\end{minipage}\par
\vspace{28pt}

\noindent\begin{minipage}{\linewidth}
  \centering
  \begin{minipage}[t]{.48\linewidth}
  \captionsetup{type=figure}
  \centering
  \resizebox{.82\linewidth}{!}{%
    \begingroup
\definecolor{bonusblue}{HTML}{24576E}
\definecolor{acgold}{HTML}{FFE8A3}
\begin{tikzpicture}[
  team/.style={circle,draw=black,fill=white,line width=.55pt,
    minimum size=5.2mm,inner sep=0pt,font=\small},
  ac team/.style={team,fill=acgold,double,double distance=.7pt,line width=.45pt},
  match/.style={-{Stealth[length=1.9mm,width=1.3mm]},line width=.6pt},
  reversed match/.style={match,line width=1.1pt,draw=bonusblue},
  every node/.style={font=\small}
]
\path[use as bounding box] (0,0) rectangle (7.4,4.45);

\node[anchor=west,inner sep=0pt,font=\bfseries\small] at (.1,4.25) {Case (f$_3$)};
\node[anchor=east,inner sep=0pt,font=\footnotesize] at (7.3,4.25) {$|H(T)|=2,\quad |H(T')|=1$};
\node[ac team] (i) at (1.500,2.070) {$\boldsymbol{i}$};
\node[team] (j) at (5.900,2.070) {$j$};
\node[ac team] (h) at (1.500,3.435) {$\boldsymbol{h}$};
\node[team] (ell) at (5.900,3.435) {$\ell$};
\draw[match] (ell) -- (h);
\draw[match] (h) -- (i);
\draw[reversed match] (i) -- (j);
\node[anchor=north,inner sep=0pt,font=\fontsize{8.5}{10}\selectfont] at (3.7,1.5) {
\setlength{\tabcolsep}{3pt}
\renewcommand{\arraystretch}{1.0}
\begin{tabular}{@{}p{1.70cm}>{\columncolor{acgold}\centering\arraybackslash}p{1.08cm}>{\centering\arraybackslash}p{1.08cm}>{\columncolor{acgold}\centering\arraybackslash}p{1.08cm}>{\centering\arraybackslash}p{1.08cm}@{}}
\toprule
\textbf{Team} & $\boldsymbol{i}$ & $j$ & $\boldsymbol{h}$ & $\ell$ \\
\midrule
\textbf{Wins in $T$} & $n-2$ & \textemdash & $n-2$ & \textemdash \\
\textbf{Bonus} & \textemdash & \textemdash & $\color{bonusblue}P\to0$ & $\color{bonusblue}Q\to B$ \\
\bottomrule
\end{tabular}
};
\end{tikzpicture}
\endgroup%
  }
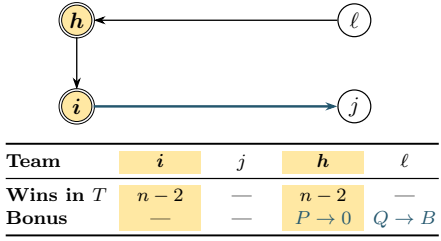
  \caption{Case (f$_3$), $2\to1$: the external blocker $\ell$ lies outside $\{i,j\}$.}
  \label{fig:bonus-transition-f3}
  \end{minipage}\hfill
  \begin{minipage}[t]{.48\linewidth}
  \captionsetup{type=figure}
  \centering
  \resizebox{.82\linewidth}{!}{%
    \begingroup
\definecolor{bonusblue}{HTML}{24576E}
\definecolor{acgold}{HTML}{FFE8A3}
\begin{tikzpicture}[
  team/.style={circle,draw=black,fill=white,line width=.55pt,
    minimum size=5.2mm,inner sep=0pt,font=\small},
  ac team/.style={team,fill=acgold,double,double distance=.7pt,line width=.45pt},
  match/.style={-{Stealth[length=1.9mm,width=1.3mm]},line width=.6pt},
  reversed match/.style={match,line width=1.1pt,draw=bonusblue},
  every node/.style={font=\small}
]
\path[use as bounding box] (0,0) rectangle (7.4,4.45);

\node[anchor=west,inner sep=0pt,font=\bfseries\small] at (.1,4.25) {Case (g$_1$)};
\node[anchor=east,inner sep=0pt,font=\footnotesize] at (7.3,4.25) {$|H(T)|=2,\quad |H(T')|=2$};
\node[ac team] (i) at (2.700,2.070) {$\boldsymbol{i}$};
\node[team] (j) at (6.200,2.070) {$j$};
\node[ac team] (h) at (4.450,3.435) {$\boldsymbol{h}$};
\node[team] (ell) at (0.900,3.045) {$\ell$};
\draw[match] (ell) -- (i);
\draw[match] (i) -- (h);
\draw[match] (h) -- (j);
\draw[reversed match] (i) -- (j);
\node[anchor=north,inner sep=0pt,font=\fontsize{8.5}{10}\selectfont] at (3.7,1.5) {
\setlength{\tabcolsep}{3pt}
\renewcommand{\arraystretch}{1.0}
\begin{tabular}{@{}p{1.70cm}>{\columncolor{acgold}\centering\arraybackslash}p{1.08cm}>{\centering\arraybackslash}p{1.08cm}>{\columncolor{acgold}\centering\arraybackslash}p{1.08cm}>{\centering\arraybackslash}p{1.08cm}@{}}
\toprule
\textbf{Team} & $\boldsymbol{i}$ & $j$ & $\boldsymbol{h}$ & $\ell$ \\
\midrule
\textbf{Wins in $T$} & $n-2$ & $n-3$ & $n-2$ & \textemdash \\
\textbf{Bonus} & $\color{bonusblue}P\to Q$ & \textemdash & $\color{bonusblue}0\to P$ & $\color{bonusblue}Q\to0$ \\
\bottomrule
\end{tabular}
};
\end{tikzpicture}
\endgroup%
  }
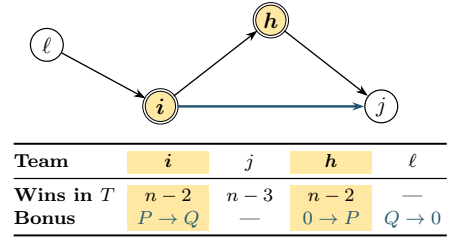
  \caption{Case (g$_1$), $2\to2$: the bonus of $i$ changes from $P$ to $Q$.}
  \label{fig:bonus-transition-g1}
  \end{minipage}
\end{minipage}\par
\vspace{28pt}

\noindent\begin{minipage}{\linewidth}
  \centering
  \begin{minipage}[t]{.48\linewidth}
  \captionsetup{type=figure}
  \centering
  \resizebox{.82\linewidth}{!}{%
    \begingroup
\definecolor{bonusblue}{HTML}{24576E}
\definecolor{acgold}{HTML}{FFE8A3}
\begin{tikzpicture}[
  team/.style={circle,draw=black,fill=white,line width=.55pt,
    minimum size=5.2mm,inner sep=0pt,font=\small},
  ac team/.style={team,fill=acgold,double,double distance=.7pt,line width=.45pt},
  match/.style={-{Stealth[length=1.9mm,width=1.3mm]},line width=.6pt},
  reversed match/.style={match,line width=1.1pt,draw=bonusblue},
  every node/.style={font=\small}
]
\path[use as bounding box] (0,0) rectangle (7.4,4.45);

\node[anchor=west,inner sep=0pt,font=\bfseries\small] at (.1,4.25) {Case (g$_2$)};
\node[anchor=east,inner sep=0pt,font=\footnotesize] at (7.3,4.25) {$|H(T)|=2,\quad |H(T')|=2$};
\node[ac team] (i) at (1.200,2.070) {$\boldsymbol{i}$};
\node[team] (j) at (4.500,2.070) {$j$};
\node[ac team] (h) at (2.850,3.435) {$\boldsymbol{h}$};
\node[team] (ell) at (6.400,3.240) {$\ell$};
\draw[match] (h) -- (i);
\draw[reversed match] (i) -- (j);
\draw[match] (j) -- (h);
\draw[match] (ell) -- (j);
\node[anchor=north,inner sep=0pt,font=\fontsize{8.5}{10}\selectfont] at (3.7,1.5) {
\setlength{\tabcolsep}{3pt}
\renewcommand{\arraystretch}{1.0}
\begin{tabular}{@{}p{1.70cm}>{\columncolor{acgold}\centering\arraybackslash}p{1.08cm}>{\centering\arraybackslash}p{1.08cm}>{\columncolor{acgold}\centering\arraybackslash}p{1.08cm}>{\centering\arraybackslash}p{1.08cm}@{}}
\toprule
\textbf{Team} & $\boldsymbol{i}$ & $j$ & $\boldsymbol{h}$ & $\ell$ \\
\midrule
\textbf{Wins in $T$} & $n-2$ & $n-3$ & $n-2$ & \textemdash \\
\textbf{Bonus} & \textemdash & $\color{bonusblue}Q\to P$ & $\color{bonusblue}P\to0$ & $\color{bonusblue}0\to Q$ \\
\bottomrule
\end{tabular}
};
\end{tikzpicture}
\endgroup%
  }
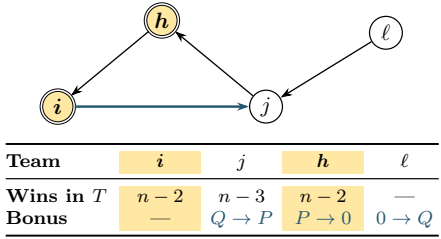
  \caption{Case (g$_2$), $2\to2$: the bonus of $j$ changes from $Q$ to $P$.}
  \label{fig:bonus-transition-g2}
  \end{minipage}\hfill
  \begin{minipage}[t]{.48\linewidth}
  \captionsetup{type=figure}
  \centering
  \resizebox{.82\linewidth}{!}{%
    \begingroup
\definecolor{bonusblue}{HTML}{24576E}
\definecolor{acgold}{HTML}{FFE8A3}
\begin{tikzpicture}[
  team/.style={circle,draw=black,fill=white,line width=.55pt,
    minimum size=5.2mm,inner sep=0pt,font=\small},
  ac team/.style={team,fill=acgold,double,double distance=.7pt,line width=.45pt},
  match/.style={-{Stealth[length=1.9mm,width=1.3mm]},line width=.6pt},
  reversed match/.style={match,line width=1.1pt,draw=bonusblue},
  every node/.style={font=\small}
]
\path[use as bounding box] (0,0) rectangle (7.4,4.45);

\node[anchor=west,inner sep=0pt,font=\bfseries\small] at (.1,4.25) {Case (g$_3$)};
\node[anchor=east,inner sep=0pt,font=\footnotesize] at (7.3,4.25) {$|H(T)|=2,\quad |H(T')|=2$};
\node[ac team] (i) at (1.500,2.070) {$\boldsymbol{i}$};
\node[team] (j) at (5.900,2.070) {$j$};
\node[ac team] (h) at (1.500,3.435) {$\boldsymbol{h}$};
\node[team] (ell) at (5.900,3.435) {$\ell$};
\draw[match] (ell) -- (h);
\draw[match] (h) -- (i);
\draw[match] (h) -- (j);
\draw[reversed match] (i) -- (j);
\node[anchor=north,inner sep=0pt,font=\fontsize{8.5}{10}\selectfont] at (3.7,1.5) {
\setlength{\tabcolsep}{3pt}
\renewcommand{\arraystretch}{1.0}
\begin{tabular}{@{}p{1.70cm}>{\columncolor{acgold}\centering\arraybackslash}p{1.08cm}>{\centering\arraybackslash}p{1.08cm}>{\columncolor{acgold}\centering\arraybackslash}p{1.08cm}>{\centering\arraybackslash}p{1.08cm}@{}}
\toprule
\textbf{Team} & $\boldsymbol{i}$ & $j$ & $\boldsymbol{h}$ & $\ell$ \\
\midrule
\textbf{Wins in $T$} & $n-2$ & $n-3$ & $n-2$ & \textemdash \\
\textbf{Bonus} & \textemdash & \textemdash & $\color{bonusblue}P\to P$ & $\color{bonusblue}Q\to Q$ \\
\bottomrule
\end{tabular}
};
\end{tikzpicture}
\endgroup%
  }
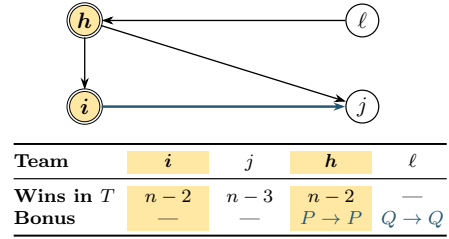
  \caption{Case (g$_3$), $2\to2$: the coalition members' bonuses are unchanged.}
  \label{fig:bonus-transition-g3}
  \end{minipage}
\end{minipage}\par
\vspace{28pt}

\noindent\begin{minipage}{\linewidth}
  \centering
  \begin{minipage}[t]{.48\linewidth}
  \captionsetup{type=figure}
  \centering
  \resizebox{.82\linewidth}{!}{%
    \begingroup
\definecolor{bonusblue}{HTML}{24576E}
\definecolor{acgold}{HTML}{FFE8A3}
\begin{tikzpicture}[
  team/.style={circle,draw=black,fill=white,line width=.55pt,
    minimum size=5.2mm,inner sep=0pt,font=\small},
  ac team/.style={team,fill=acgold,double,double distance=.7pt,line width=.45pt},
  match/.style={-{Stealth[length=1.9mm,width=1.3mm]},line width=.6pt},
  reversed match/.style={match,line width=1.1pt,draw=bonusblue},
  every node/.style={font=\small}
]
\path[use as bounding box] (0,0) rectangle (7.4,4.45);

\node[anchor=west,inner sep=0pt,font=\bfseries\small] at (.1,4.25) {Case (h)};
\node[anchor=east,inner sep=0pt,font=\footnotesize] at (7.3,4.25) {$|H(T)|=2,\quad |H(T')|=3$};
\node[team] (i) at (0.900,2.720) {$i$};
\node[team] (j) at (3.300,2.720) {$j$};
\node[ac team] (h) at (6.000,3.500) {$\boldsymbol{h}$};
\node[ac team] (ell) at (6.000,1.940) {$\boldsymbol{\ell}$};
\draw[reversed match] (i) -- (j);
\draw[match] (h) -- (j);
\draw[match] (j) -- (ell);
\draw[match] (ell) -- (h);
\node[anchor=north,inner sep=0pt,font=\fontsize{8.5}{10}\selectfont] at (3.7,1.5) {
\setlength{\tabcolsep}{3pt}
\renewcommand{\arraystretch}{1.0}
\begin{tabular}{@{}p{1.70cm}>{\centering\arraybackslash}p{1.08cm}>{\centering\arraybackslash}p{1.08cm}>{\columncolor{acgold}\centering\arraybackslash}p{1.08cm}>{\columncolor{acgold}[3pt][0pt]\centering\arraybackslash}p{1.08cm}@{}}
\toprule
\textbf{Team} & $i$ & $j$ & $\boldsymbol{h}$ & $\boldsymbol{\ell}$ \\
\midrule
\textbf{Wins in $T$} & \textemdash & $n-3$ & $n-2$ & $n-2$ \\
\textbf{Bonus} & \textemdash & $\color{bonusblue}Q\to R$ & $\color{bonusblue}0\to R$ & $\color{bonusblue}P\to R$ \\
\bottomrule
\end{tabular}
};
\end{tikzpicture}
\endgroup%
  }
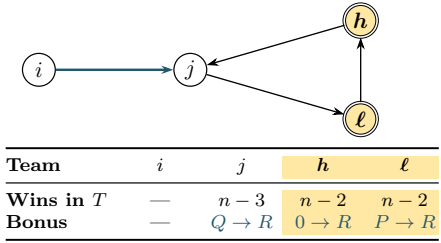
  \caption{Case (h), $2\to3$: the bonus of $j$ changes from $Q$ to $R$.}
  \label{fig:bonus-transition-h}
  \end{minipage}\hfill
  \begin{minipage}[t]{.48\linewidth}
  \captionsetup{type=figure}
  \centering
  \resizebox{.82\linewidth}{!}{%
    \begingroup
\definecolor{bonusblue}{HTML}{24576E}
\definecolor{acgold}{HTML}{FFE8A3}
\begin{tikzpicture}[
  team/.style={circle,draw=black,fill=white,line width=.55pt,
    minimum size=5.2mm,inner sep=0pt,font=\small},
  ac team/.style={team,fill=acgold,double,double distance=.7pt,line width=.45pt},
  match/.style={-{Stealth[length=1.9mm,width=1.3mm]},line width=.6pt},
  reversed match/.style={match,line width=1.1pt,draw=bonusblue},
  every node/.style={font=\small}
]
\path[use as bounding box] (0,0) rectangle (7.4,4.45);

\node[anchor=west,inner sep=0pt,font=\bfseries\small] at (.1,4.25) {Case (i)};
\node[anchor=east,inner sep=0pt,font=\footnotesize] at (7.3,4.25) {$|H(T)|=3,\quad |H(T')|=2$};
\node[ac team] (i) at (3.400,3.045) {$\boldsymbol{i}$};
\node[team] (j) at (6.300,3.598) {$j$};
\node[ac team] (h) at (1.300,1.973) {$\boldsymbol{h}$};
\node[ac team] (ell) at (5.500,1.973) {$\boldsymbol{\ell}$};
\draw[reversed match] (i) -- (j);
\draw[match] (i) -- (h);
\draw[match] (h) -- (ell);
\draw[match] (ell) -- (i);
\node[anchor=north,inner sep=0pt,font=\fontsize{8.5}{10}\selectfont] at (3.7,1.5) {
\setlength{\tabcolsep}{3pt}
\renewcommand{\arraystretch}{1.0}
\begin{tabular}{@{}p{1.70cm}>{\columncolor{acgold}\centering\arraybackslash}p{1.08cm}>{\centering\arraybackslash}p{1.08cm}>{\columncolor{acgold}\centering\arraybackslash}p{1.08cm}>{\columncolor{acgold}[3pt][0pt]\centering\arraybackslash}p{1.08cm}@{}}
\toprule
\textbf{Team} & $\boldsymbol{i}$ & $j$ & $\boldsymbol{h}$ & $\boldsymbol{\ell}$ \\
\midrule
\textbf{Wins in $T$} & $n-2$ & \textemdash & $n-2$ & $n-2$ \\
\textbf{Bonus} & $\color{bonusblue}R\to Q$ & \textemdash & $\color{bonusblue}R\to P$ & $\color{bonusblue}R\to0$ \\
\bottomrule
\end{tabular}
};
\end{tikzpicture}
\endgroup%
  }
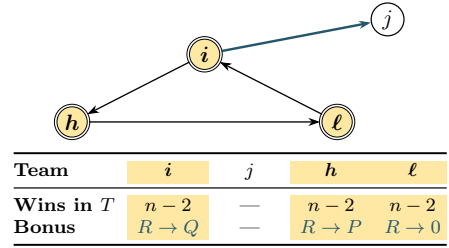
  \caption{Case (i), $3\to2$: the bonus of $i$ changes from $R$ to $Q$.}
  \label{fig:bonus-transition-i}
  \end{minipage}
\end{minipage}\par

\clearpage
\endgroup

\section{Missing Proofs for Section~\ref{sec:2NM}}
\label{app:detailed-case-analysis}

\newenvironment{restatedlemma}[1]
  {\par\medskip\begingroup
   \noindent\textbf{Lemma~\ref{#1}.}\enspace\itshape}
  {\par\endgroup\medskip}

\subsection{Well-Definedness}\label{app:welldefine}

\begin{restatedlemma}{lem:rule2-well-defined}
Rule~\ref{rule:2NM} defines a probability distribution on every tournament.
\end{restatedlemma}

\begin{proof}
The claim is immediate when a Condorcet winner exists. Suppose that \(T\) has no Condorcet winner. We will prove \(\sum_{i\in N}t_i(T)\le M\), which makes the remaining probability mass nonnegative. We first bound the sum of the unbonused win-strength scores.

Let \(\ell_i(T)=n-1-d_i(T)\) denote the number of losses of team \(i\), and set $x_i=\ell_i(T)-1=n-2-d_i(T)$.
Then \(x_i\ge 0\) and \(s_i(T)=3^{-x_i}\). In the sum of the unbonused win-strength scores, \(s_i(T)\) appears with coefficient \(3/2\) in \(i\)'s own score and once in the score of each team that defeats \(i\). Hence
\begin{equation}
\label{eq:total-unbonused-score}
    \sum_{i\in N} t_i^0(T)
    =
    \sum_{i\in N}
    \left(\frac32+\ell_i(T)\right)s_i(T)
    =
    \sum_{i\in N}\frac{x_i+5/2}{3^{x_i}}.
\end{equation}

The sequence $f(x)=\frac{x+5/2}{3^x}$ is discrete convex, since
\[
    f(x)-2f(x+1)+f(x+2)
    =
    \frac{2(2x+3)}{3^{x+2}}
    >0.
\]

By Landau’s characterization of tournament score sequences \cite{landauDominanceRelationsStructure1953}, this discrete convexity implies that, among tournaments without a Condorcet winner, the sum in \eqref{eq:total-unbonused-score} is maximized by the following sequence of win counts \(d_i(T)\):
\[
    n-2,\ n-2,\ n-2,\ n-4,\ n-5,\ldots,1,0.
\]
The corresponding scores \(s_i(T)\) are
\[
    1,\ 1,\ 1,\ 3^{-2},\ 3^{-3},\ldots,3^{-(n-2)}.
\]
Therefore
\begin{align*}
    \sum_{i\in N}t_i^0(T)
    &\le
    3\cdot\frac52
    +
    \sum_{x=2}^{n-2}\frac{x+5/2}{3^x}\\
    &<
    \frac{15}{2}
    +
    \sum_{x=2}^{\infty}\frac{x}{3^x}
    +
    \frac52\sum_{x=2}^{\infty}\frac{1}{3^x}\\
    &=
    \frac{15}{2}+\frac{5}{12}+\frac{5}{12}
    =
    \frac{25}{3}.
\end{align*}

The possible values of the total bonus are
\[
    0,\quad
    B=\frac{89}{36},\quad
    P+Q=\frac{43}{12},\quad
    3R=\frac72.
\]
Their maximum is \(P+Q=43/12\). Adding this to the bound on the unbonused scores gives
\[
    \sum_{i\in N}t_i(T)
    \le
    \frac{25}{3}+\frac{43}{12}
    =
    \frac{143}{12}
    =
    M.
\]
Thus the remaining probability mass in Rule~\ref{rule:2NM} is nonnegative, so every \(r_i(T)\) is nonnegative. Moreover,
\[
    \sum_{i\in N}r_i(T)
    =
    \frac{\sum_i t_i(T)}{M}
    +
    1-\frac{\sum_i t_i(T)}{M}
    =
    1.
\]
\end{proof}

\subsection{Properties for Small Tournaments}
\label{app:property}
\begin{restatedlemma}{lem:small-feasibility}
Assume that neither \(T\) nor \(T'\) has a Condorcet winner.
\begin{enumerate}
\item If \(|H(T)|=1\) and \(|H(T')|=0\), then \(n=4\) is impossible. Moreover, if \(n=5\), then \(v=3\) is impossible.
\item If \(|H(T)|=2\) and \(|H(T')|=1\), then \(n\ge5\).
\item In case \emph{(f\(_2\))}, one has \(n\ge v+3\).
\end{enumerate}
\end{restatedlemma}

\begin{proof}
For item~\emph{(1)}, \(i\) is the unique almost-Condorcet winner in \(T\). First suppose that \(n=4\). Let \(N=\{i,j,h,\ell\}\), and label the two teams outside \(\{i,j\}\) so that $h\to_T i$ and $i\to_T j,\ell$.
Since \(j\) does not enter \(H(T')\), both \(h\) and \(\ell\) defeat \(j\), or $h\to_T j, \ell\to_T j$. If \(h\to_T\ell\), then \(h\) is a Condorcet winner in \(T\). If \(\ell\to_T h\), then both \(h\) and \(\ell\) have \(2=n-2\) wins in \(T\), contradicting the uniqueness of \(i\).

Now suppose that \(n=5\) and \(v=3\). Let \(h\) be the unique team that defeats \(i\). Then \(i\) defeats \(j\) and the two teams in \(N\setminus\{i,j,h\}\), while all three teams outside \(\{i,j\}\) defeat \(j\). Team \(h\) already defeats \(i\) and \(j\); to avoid becoming another almost-Condorcet winner, it must lose to both teams in \(N\setminus\{i,j,h\}\). The winner of their match then defeats \(h\), \(j\), and the loser of that match, giving it \(3=n-2\) wins in \(T\). This again contradicts the uniqueness of \(i\).

\smallskip

For item~\emph{(2)}, suppose that \(n=4\), write \(H(T)=\{i,h\}\), and let \(N=\{i,j,h,\ell\}\). Since \(i\) leaves the set of almost-Condorcet winners and \(j\) does not enter it, both \(h\) and \(\ell\) defeat \(j\). If \(i\to_T h\), then \(h\) loses only to \(i\), so \(h\to_T j,\ell\), while \(\ell\to_T i\); hence \(\ell\) is also an almost-Condorcet winner. If \(h\to_T i\), then \(i\to_T j,\ell\), and \(h\) must lose to \(\ell\); again \(\ell\) is an almost-Condorcet winner. Both cases contradict \(|H(T)|=2\).

For item~\emph{(3)}, in case (f\(_2\)), \(j\) becomes the unique blocker of the almost-Condorcet winner \(h\) that remains after the reversal. Thus \(j\to_T h\), so \(h\) is not among the \(v\) teams outside \(\{i,j\}\) that defeat \(j\). There are \(n-2\) teams outside \(\{i,j\}\) in total, including \(h\). Consequently, \(v\le n-3\), or equivalently \(n\ge v+3\).
\end{proof}

\subsection{No-Condorcet Cases under Bonus Transitions}
\label{app:proof-no-cw-bonus-cases}

\begin{restatedlemma}{lem:no-cw-bonus-cases}
If neither \(T\) nor \(T'\) has a Condorcet winner, then every case in Table~\ref{tab:bonus-transitions} satisfies $G\le3L$.

\end{restatedlemma}

\begin{proof}[Proof of Lemma~\ref{lem:no-cw-bonus-cases}]
In \eqref{eq:3LminusG}, the contribution from bonuses is \(4\Delta-3n\delta_i-n\delta_j\). It is easy to verify that this contribution is nonnegative in cases (b), (e\(_1\))--(e\(_3\)), (g\(_1\)), and (h), and zero in cases (a), (d), and (g\(_3\)). Since \(3L_0-G_0>0\) by \eqref{eq:base-slack}, these cases satisfy \(G\le3L\). It remains to consider cases (c), (f\(_1\))--(f\(_3\)), (g\(_2\)), and (i).

\smallskip
\noindent\textbf{Case (c): \(1\to0\).} Here \(u=1\) and \(v\ge2\), so \eqref{eq:base-slack} and \eqref{eq:3LminusG} give
\[
    3L-G
    =
    8\left(
        \frac{n-2}{3}
        +
        \frac{v+1}{3^v}
    \right)
    -
    \frac{89}{9}.
\]
Lemma~\ref{lem:small-feasibility} excludes the parameter pairs \((n,v)=(4,2)\) and \((5,3)\). Thus
\[
    8\left(
        \frac{n-2}{3}
        +
        \frac{v+1}{3^v}
    \right)
    \ge
    \frac{32}{3},
\]
and hence
\[
    3L-G
    \ge
    \frac{32}{3}-\frac{89}{9}
    =
    \frac79
    >0.
\]

\smallskip
\noindent\textbf{Case (f\(_3\)): \(2\to1\).} The bonuses of \(i\) and \(j\) are unchanged, and Lemma~\ref{lem:small-feasibility} gives \(n\ge5\). Therefore,
\[
    3L-G
    =
    8\left(
        \frac{n-2}{3}
        +
        \frac{v+1}{3^v}
    \right)
    -
    \frac{40}{9}
    >
    8-\frac{40}{9}
    =
    \frac{32}{9}.
\]

\smallskip
\noindent\textbf{Case (f\(_2\)): \(j:Q\to B\).} Equation~\eqref{eq:3LminusG} becomes
\[
    3L-G
    =
    8\left(
        \frac{n-2}{3}
        +
        \frac{v+1}{3^v}
    \right)
    -
    \frac{40}{9}
    -
    \frac{19n}{18}.
\]
For fixed \(v\), the right-hand side increases with \(n\), since \(8/3-19/18=29/18>0\). Lemma~\ref{lem:small-feasibility} gives \(n\ge v+3\), so it suffices to bound this expression at \(n=v+3\). Substituting \(v=2\) and \(v=3\) gives \(3L-G\ge17/18\) and \(3L-G\ge29/27\), respectively. For \(v\ge4\), dropping the positive summand \(8(v+1)/3^v\) yields
\[
    3L-G
    \ge
    \frac{29v-89}{18}
    >0.
\]

\smallskip
\noindent\textbf{Case (f\(_1\)): \(i:P\to B\).} In this case,
\[
    3L-G
    =
    8\left(
        \frac{n-2}{3}
        +
        \frac{v+1}{3^v}
    \right)
    -
    \frac{40}{9}
    -
    \frac{11n}{12}.
\]
Lemma~\ref{lem:small-feasibility} gives \(n\ge5\). When \(n=5\), one has \(v\in\{2,3\}\), and direct substitution yields
\[
    3L-G\ge
    \begin{cases}
        59/36,&v=2,\\[1mm]
        17/108,&v=3.
    \end{cases}
\]
For \(n\ge6\), omitting the positive summand \(8(v+1)/3^v\) gives
\[
    3L-G
    \ge
    \frac{8(n-2)}{3}
    -
    \frac{40}{9}
    -
    \frac{11n}{12}
    =
    \frac{63(n-1)-289}{36}
    >0.
\]

\smallskip
\noindent\textbf{Case (g\(_2\)): \(2\to2\) exchange, \(j:Q\to P\).} Here \(u=v=1\), \(\delta_j=P-Q=3/4\), and \(\Delta=0\). Thus
\[
    3L-G
    =
    \frac{8n}{3}-\frac{3n}{4}
    =
    \frac{23n}{12}
    >0.
\]

\smallskip
\noindent\textbf{Case (i): \(3\to2\).} In this case, \(u=1\), \(\delta_i=Q-R=1/4\), and \(\Delta=1/12\), with \(n\ge4\). Consequently,
\begin{align*}
    3L-G
    &\ge
    \frac{8(n-2)}{3}
    +
    \frac13
    -
    \frac{3n}{4}\\
    &=
    \frac{23(n-1)-37}{12}
    >0.
\end{align*}

Case (j) is impossible by Lemma~\ref{lem:bonus-table}. Together with the cases handled at the start of the proof, the analysis above exhaust the cases in Table~\ref{tab:bonus-transitions}.
\end{proof}

\subsection{Condorcet-Creation Cases}
\label{app:proof-condorcet-creation}

\begin{restatedlemma}{lem:condorcet-creation}
Suppose \(T\) has no Condorcet winner and \(j\) becomes a Condorcet winner in \(T'\). Then $r_j(T')-r_j(T)
    \le
    3\bigl(r_i(T)-r_i(T')\bigr).$
\end{restatedlemma}

\begin{proof}[Proof of Lemma~\ref{lem:condorcet-creation}]
Since \(j\) becomes a Condorcet winner after the reversal, it has \(n-2\) wins in \(T\) and loses only to \(i\). Rule~\ref{rule:2NM} gives \(r_j(T')=1\) and \(r_i(T')=0\), so the desired inequality reduces to \(1-r_j(T)\le3r_i(T)\).

The uniform share of the remaining probability mass in Rule~\ref{rule:2NM} is nonnegative, as shown in Lemma~\ref{lem:rule2-well-defined}. Therefore,
\[
    r_i(T)\ge\frac{t_i(T)}{M},
    \quad
    r_j(T)\ge\frac{t_j(T)}{M}.
\]
Consequently,
\[
    \frac{1-r_j(T)}{r_i(T)}
    \le
    \frac{M-t_j(T)}{t_i(T)}.
\]
It suffices to bound the right-hand side by \(3\). We distinguish four cases according to \(H(T)\).

\smallskip
\noindent\textbf{Case 1: \(H(T)=\{j\}\).} Team \(i\) is the unique blocker of \(j\), so
\[
    t_i(T)\ge1+B=\frac{125}{36},
    \quad
    t_j(T)\ge\frac32.
\]
Thus
\[
    \frac{M-t_j(T)}{t_i(T)}
    \le
    \frac{M-\frac32}{1+B}
    =
    3.
\]

\smallskip
\noindent\textbf{Case 2: \(H(T)=\{j,h\}\) with \(i\ne h\).} Since \(j\) loses only to \(i\), one has \(j\to_T h\). Thus \(j\) is the internal blocker of \(h\), whereas \(i\) is the external blocker of \(j\). These roles give
\[
    t_j(T)\ge\frac32+1+P=\frac{14}{3},
    \quad
    t_i(T)\ge1+Q=\frac{29}{12},
\]
and hence
\[
    \frac{M-t_j(T)}{t_i(T)}
    \le
    \frac{M-\frac{14}{3}}{\frac{29}{12}}
    =
    3.
\]

\smallskip
\noindent\textbf{Case 3: \(H(T)=\{i,j\}\).} Here \(i\) is the internal blocker of \(j\), giving
\[
    t_i(T)\ge\frac32+1+P=\frac{14}{3},
    \quad
    t_j(T)\ge\frac32.
\]
It follows that
\[
    \frac{M-t_j(T)}{t_i(T)}
    \le
    \frac{M-\frac32}{\frac{14}{3}}
    =
    \frac{125}{56}
    <3.
\]

\smallskip
\noindent\textbf{Case 4: \(|H(T)|=3\).} The three almost-Condorcet winners form a directed cycle. In \(t_i(T)\) and \(t_j(T)\), the terms \(\frac32s_i(T)\) and \(\frac32s_j(T)\), respectively, each contribute \(\frac32\). Each of \(i\) and \(j\) also defeats another member of the cycle, adding \(1\) to its unbonused win-strength score, and receives bonus \(R\). Hence
\[
    t_i(T),t_j(T)
    \ge
    \frac32+1+R
    =
    \frac{11}{3},
\]
and
\[
    \frac{M-t_j(T)}{t_i(T)}
    \le
    \frac{M-\frac{11}{3}}{\frac{11}{3}}
    =
    \frac94
    <3.
\]
The four cases give \(\frac{M-t_j(T)}{t_i(T)}\le3\), completing the proof.
\end{proof}

\subsection{Monotonicity under Bonus Transitions}
\label{app:proof-bonus-monotonicity}

\begin{restatedlemma}{lem:bonus-monotonicity}
If neither \(T\) nor \(T'\) has a Condorcet winner, then every transition in Table~\ref{tab:bonus-transitions} satisfies
$r_i(T)\ge r_i(T').$
\end{restatedlemma}

\begin{proof}[Proof of Lemma~\ref{lem:bonus-monotonicity}]
Since \(L=nM\bigl(r_i(T)-r_i(T')\bigr)\), the target is \(L\ge0\). By \eqref{eq:master-bonus},
\[
    L=L_0-n\delta_i+\Delta.
\]
Only cases (c), (f\(_1\))--(f\(_3\)), (h), and (i) can have a negative contribution \(-n\delta_i+\Delta\) from the bonuses. Whenever \(u=1\), \eqref{eq:L0} gives
\[
    L_0
    =
    n-1-\frac13+\frac{n+2+2v}{3^v}
    \ge
    n-1-\frac13.
\]

\smallskip
\noindent\textbf{Case (c).} Lemma~\ref{lem:small-feasibility} gives \(n\ge5\), so
\[
    L=L_0-B
    \ge
    \frac{11}{3}-\frac{89}{36}
    =
    \frac{43}{36}
    >0.
\]

\smallskip
\noindent\textbf{Cases (f\(_2\)) and (f\(_3\)).} The estimate for \(L_0\) also gives
\[
    L=L_0-\frac{10}{9}>0.
\]

\smallskip
\noindent\textbf{Case (f\(_1\)).} Using \(n\ge5\) from Lemma~\ref{lem:small-feasibility}, we obtain
\[
    L
    =
    L_0-\frac{11n}{36}-\frac{10}{9}
    \ge
    \frac{25(n-1)-63}{36}
    >0.
\]

\smallskip
\noindent\textbf{Case (h).} Here \(v=1\), and \eqref{eq:L0} gives
\[
    L_0
    \ge
    \frac{n+4}{3}.
\]
Therefore,
\[
    L=L_0-\frac{1}{12}>0.
\]

\smallskip
\noindent\textbf{Case (i).} Finally,
\begin{align*}
    L
    &=
    L_0-\frac{n}{4}+\frac{1}{12}\\
    &\ge
    n-1-\frac13-\frac{n}{4}+\frac{1}{12}\\
    &=
    \frac{3(n-1)}{4}-\frac12
    >0.
\end{align*}
In every other transition, the bonuses either leave \(L\) equal to \(L_0\) or increase it above \(L_0\). Since \(L_0>0\), these transitions also satisfy \(L\ge0\), completing the proof.
\end{proof}

\end{document}